\documentclass{article}

\usepackage{amsmath, amsthm, amssymb}
\usepackage{graphicx}
\usepackage{enumerate}
\usepackage{comment} 
\usepackage{tikz}

\newtheorem{theorem}{Theorem}
\newtheorem{lemma}[theorem]{Lemma}

\theoremstyle{definition}

\newtheorem{definition}{Definition}
\newtheorem{property}{Property}

\newcommand{\pats}{{\sc Pats}}
\newcommand{\oneinthreeSAT}{\ensuremath{\mbox{{\sc 1-in-3-Sat}}}}
\newcommand{\sat}{\ensuremath{\mbox{\sc Sat}}}

\newcommand{\NP}{{\bf NP}}

\newcommand{\north}{{\tt N}}
\newcommand{\west}{{\tt W}}
\newcommand{\south}{{\tt S}}
\newcommand{\east}{{\tt E}}

\newcommand{\true}{{\tt T}}
\newcommand{\false}{{\tt F}}

\newcommand{\ce}{{\tt CE}}

\pgfrealjobname{11pats_arxiv}

\tikzstyle{tile} = [draw=black, minimum width=7mm, minimum height=7mm, node distance=7mm]

\tikzstyle{bluetile} = [tile, fill=cyan!20]
\tikzstyle{graytile} = [tile, fill=gray!20]
\tikzstyle{yellowtile} = [tile, fill=yellow]

\newcommand{\blacklowertriangle}[2]{
	\draw[fill=black] (#1, #2) -- (#1-0.35, #2-0.35) -- (#1+0.35, #2-0.35) -- (#1+0.35, #2+0.35) -- cycle; 
}

\newcommand{\whitelowertriangle}[2]{
	\draw[fill=white] (#1, #2) -- (#1-0.35, #2-0.35) -- (#1+0.35, #2-0.35) -- (#1+0.35, #2+0.35) -- cycle; 
}

\newcommand{\glnF}[2]{
	\draw[thick] (#1-0.1, #2+0.35) -- (#1-0.1, #2+0.65);
	\node at (#1+0.025, #2+0.5) {\tiny F}
}

\newcommand{\glnT}[2]{
	\draw[thick] (#1-0.2, #2+0.35) -- (#1-0.2, #2+0.65);
	\node at (#1-0.075, #2+0.5) {\tiny T}
}

\newcommand{\glnn}[2]{
	\draw[thick] (#1+0.2, #2+0.35) -- (#1+0.2, #2+0.65);
	\node at (#1+0.325, #2+0.5) {\scriptsize n}
}

\newcommand{\glnv}[2]{
	\draw[thick] (#1+0.1, #2+0.35) -- (#1+0.1, #2+0.65);
	\node at (#1+0.225, #2+0.5) {\scriptsize v}
}

\newcommand{\glnc}[2]{
	\draw[thick] (#1, #2+0.35) -- (#1, #2+0.65);
	\node at (#1+0.125, #2+0.5) {\scriptsize c}
}

\newcommand{\glwF}[2]{
	\draw[thick] (#1-0.35, #2+0.1) -- (#1-0.65, #2+0.1);
	\node at (#1-0.5, #2+0.2) {\tiny F}
}

\newcommand{\glwT}[2]{
	\draw[thick] (#1-0.35, #2+0.2) -- (#1-0.65, #2+0.2);
	\node at (#1-0.5, #2+0.3) {\tiny T}
}

\newcommand{\glwf}[2]{
	\draw[thick] (#1-0.35, #2-0.2) -- (#1-0.65, #2-0.2);
	\node at (#1-0.5, #2-0.1) {\scriptsize f}
}

\newcommand{\glwt}[2]{
	\draw[thick] (#1-0.35, #2-0.1) -- (#1-0.65, #2-0.1);
	\node at (#1-0.5, #2) {\scriptsize t}
}

\newcommand{\glws}[2]{
	\draw[thick] (#1-0.35, #2) -- (#1-0.65, #2);
	\node at (#1-0.5, #2+0.1) {\scriptsize s}
}

\newcommand{\glsF}[2]{
	\draw[thick] (#1-0.1, #2-0.35) -- (#1-0.1, #2-0.65);
	\node at (#1+0.025, #2-0.5) {\tiny F}
}

\newcommand{\glsT}[2]{
	\draw[thick] (#1-0.2, #2-0.35) -- (#1-0.2, #2-0.65);
	\node at (#1-0.075, #2-0.5) {\tiny T}
}

\newcommand{\glsn}[2]{
	\draw[thick] (#1+0.2, #2-0.35) -- (#1+0.2, #2-0.65);
	\node at (#1+0.325, #2-0.5) {\scriptsize n}
}

\newcommand{\glsv}[2]{
	\draw[thick] (#1+0.1, #2-0.35) -- (#1+0.1, #2-0.65);
	\node at (#1+0.225, #2-0.5) {\scriptsize v}
}

\newcommand{\glsc}[2]{
	\draw[thick] (#1, #2-0.35) -- (#1, #2-0.65);
	\node at (#1+0.125, #2-0.5) {\scriptsize c}
}

\newcommand{\gleF}[2]{
	\draw[thick] (#1+0.35, #2+0.1) -- (#1+0.65, #2+0.1);
	\node at (#1+0.5, #2+0.2) {\tiny F}
}

\newcommand{\gleT}[2]{
	\draw[thick] (#1+0.35, #2+0.2) -- (#1+0.65, #2+0.2);
	\node at (#1+0.5, #2+0.3) {\tiny T}
}

\newcommand{\glef}[2]{
	\draw[thick] (#1+0.35, #2-0.2) -- (#1+0.65, #2-0.2);
	\node at (#1+0.5, #2-0.1) {\scriptsize f}
}

\newcommand{\glet}[2]{
	\draw[thick] (#1+0.35, #2-0.1) -- (#1+0.65, #2-0.1);
	\node at (#1+0.5, #2) {\scriptsize t}
}

\newcommand{\gles}[2]{
	\draw[thick] (#1+0.35, #2) -- (#1+0.65, #2);
	\node at (#1+0.5, #2+0.1) {\scriptsize s}
}

\newcommand{\blueff}[2]{
	\node[bluetile] at (#1, #2) {};
	\glnF{#1}{#2}; 
	\glwF{#1}{#2};
	\glsF{#1}{#2};
	\gleF{#1}{#2}
}

\newcommand{\blueft}[2]{
	\node[bluetile] at (#1, #2) {};
	\glnT{#1}{#2};
	\glwF{#1}{#2};
	\glsT{#1}{#2};
	\gleF{#1}{#2}
}

\newcommand{\bluetf}[2]{
	\node[bluetile] at (#1, #2) {};
	\glnF{#1}{#2}; 
	\glwT{#1}{#2};
	\glsF{#1}{#2};
	\gleT{#1}{#2}
}

\newcommand{\bluett}[2]{
	\node[bluetile] at (#1, #2) {};
	\glnT{#1}{#2};
	\glwT{#1}{#2};
	\glsT{#1}{#2};
	\gleT{#1}{#2}
}

\newcommand{\whitef}[2]{
	\node[tile] at (#1, #2) {};
	\glnn{#1}{#2};
	\glwf{#1}{#2};
	\glsn{#1}{#2};
	\glef{#1}{#2}
}

\newcommand{\whitet}[2]{
	\node[tile] at (#1, #2) {};
	\glnn{#1}{#2};
	\glwt{#1}{#2};
	\glsn{#1}{#2};
	\glet{#1}{#2}
}

\newcommand{\blackf}[2]{
	\node[tile, fill=black] at (#1, #2) {};
	\glnv{#1}{#2};
	\glwf{#1}{#2};
	\glsv{#1}{#2};
	\glef{#1}{#2}
}

\newcommand{\blackt}[2]{
	\node[tile, fill=black] at (#1, #2) {};
	\glnv{#1}{#2};
	\glwt{#1}{#2};
	\glsv{#1}{#2};
	\glet{#1}{#2}
}

\newcommand{\DGNLwF}[2]{
	\node[bluetile] at (#1, #2) {};
	\draw[fill=white] (#1+0.35, #2-0.35) -- (#1+0.35, #2+0.35) -- (#1-0.35, #2-0.35) -- cycle; 
	\glnF{#1}{#2};
	\glwF{#1}{#2};
	\glsn{#1}{#2};
	\glef{#1}{#2}
}

\newcommand{\DGNLwT}[2]{
	\node[bluetile] at (#1, #2) {};
	\draw[fill=white] (#1+0.35, #2-0.35) -- (#1+0.35, #2+0.35) -- (#1-0.35, #2-0.35) -- cycle; 
	\glnF{#1}{#2};
	\glwT{#1}{#2};
	\glsn{#1}{#2};
	\glet{#1}{#2}
}

\newcommand{\DGNLbF}[2]{
	\node[bluetile] at (#1, #2) {};
	\draw[fill=black] (#1+0.35, #2-0.35) -- (#1+0.35, #2+0.35) -- (#1-0.35, #2-0.35) -- cycle; 
	\glnF{#1}{#2};
	\glwF{#1}{#2};
	\glsv{#1}{#2};
	\glef{#1}{#2}
}

\newcommand{\DGNLbT}[2]{
	\node[bluetile] at (#1, #2) {};
	\draw[fill=black] (#1+0.35, #2-0.35) -- (#1+0.35, #2+0.35) -- (#1-0.35, #2-0.35) -- cycle; 
	\glnT{#1}{#2};
	\glwT{#1}{#2};
	\glsv{#1}{#2};
	\glet{#1}{#2}
}

\newcommand{\Initf}[2]{
	\node[graytile] at (#1, #2) {};
	\node at (#1, #2) {\scriptsize {\tt Init}}; 
	\glnc{#1}{#2};
	\glwf{#1}{#2};
	\glsc{#1}{#2};
	\gleF{#1}{#2};
}

\newcommand{\Initt}[2]{
	\node[graytile] at (#1, #2) {};
	\node at (#1, #2) {\scriptsize {\tt Init}}; 
	\glnc{#1}{#2};
	\glwt{#1}{#2};
	\glsc{#1}{#2};
	\gleT{#1}{#2};
}

\newcommand{\redF}[2]{
	\node[tile, fill=red, text=white] at (#1, #2) {\tt F}; 
	\glnc{#1}{#2};
	\glwF{#1}{#2};
	\glsc{#1}{#2};
	\glef{#1}{#2};
}

\newcommand{\blueT}[2]{
	\node[tile, fill=blue, text=white] at (#1, #2) {\tt T}; 
	\glnc{#1}{#2};
	\glwT{#1}{#2};
	\glsc{#1}{#2};
	\glet{#1}{#2};
}

\newcommand{\yellows}[2]{
	\node[tile, fill=yellow] at (#1, #2) {};
	\glnT{#1}{#2};
	\glws{#1}{#2};
	\glsT{#1}{#2};
	\gles{#1}{#2};
}

\newcommand{\greenSat}[2]{
	\node[tile, fill=green] at (#1, #2) {\footnotesize {\tt Sat}};
	\glnF{#1}{#2};
	\glws{#1}{#2};
	\glsc{#1}{#2};
	\gleF{#1}{#2};
}

\newcommand{\CEff}[2]{
	\node[tile] at (#1, #2) {\tt CE};
	\glnF{#1}{#2};
	\glwf{#1}{#2};
	\glsF{#1}{#2};
	\glef{#1}{#2};
}

\newcommand{\CEfs}[2]{
	\node[tile] at (#1, #2) {\tt CE};
	\glnF{#1}{#2};
	\glwf{#1}{#2};
	\glsT{#1}{#2};
	\gles{#1}{#2};
}

\newcommand{\CEss}[2]{
	\node[tile] at (#1, #2) {\tt CE};
	\glnF{#1}{#2};
	\glws{#1}{#2};
	\glsF{#1}{#2};
	\gles{#1}{#2};
}

\newcommand{\bluetile}[7]{
	\node[bluetile] at (#1, #2) {#3}; 
	\node at (#1, #2+0.5) {#4};
	\node at (#1-0.5, #2) {#5};
	\node at (#1, #2-0.5) {#6};
	\node at (#1+0.5, #2) {#7};
}

\newcommand{\yellowtile}[7]{
	\node[bluetile, fill=yellow] at (#1, #2) {#3}; 
	\node at (#1, #2+0.5) {\scriptsize #4};
	\node at (#1-0.5, #2) {\scriptsize #5};
	\node at (#1, #2-0.5) {\scriptsize #6};
	\node at (#1+0.5, #2) {\scriptsize #7};
}

\newcommand{\whitetile}[7]{
	\node[tile] at (#1, #2) {#3}; 
	\node at (#1, #2+0.5) {\scriptsize #4};
	\node at (#1-0.5, #2) {\scriptsize #5};
	\node at (#1, #2-0.5) {\scriptsize #6};
	\node at (#1+0.5, #2) {\scriptsize #7};
}

\newcommand{\seed}[2]{
	\foreach \x in {0, 1, ..., #1} {\node[tile,fill=gray] at (\x, 0) {};}
	\foreach \y in {1, ..., #2} {\node[tile,fill=gray] at (0, \y) {};}
}

\title{A manually-checkable proof for the NP-hardness of 11-color pattern self-assembly tile set synthesis\footnote{This work is supported in part by NSF Grants CCF-1049899 and CCF-1217770 to A.~J.~and M-Y.~K.~and HIIT Pump Priming Grants No.~902184/T30606 and Academy of Finland, Postdoctoral Researcher Grants 13266670/T30606 to S.~S.}}
\author{
Aleck Johnsen\footnote{Department of Electrical Engineering and Computer Science, Northwestern University, Ford Motor Company Engineering Design Center, 2133 Sheridan Road, Evanston, Illinois, 60208, USA. {\tt aleckjohnsen2012@u.northwestern.edu}}, 
Ming-Yang Kao\footnote{Department of Electrical Engineering and Computer Science, Northwestern University, 2145 Sheridan Road, Evanston, Illinois, 60208, USA. {\tt kao@northwestern.edu}}, 
and 
Shinnosuke Seki\footnote{Helsinki Institute for Information Technology (HIIT), Department of Information and Computer Science, Aalto University, P.O.Box 15400, FI-00076, Aalto, Finland. {\tt shinnosuke.seki@aalto.fi}} \footnote{Corresponding author}
}

\begin{document}

\maketitle

\begin{abstract}
	Patterned self-assembly tile set synthesis (\pats) aims at finding a minimum tile set to uniquely self-assemble a given rectangular (color) pattern. 
	For $k \ge 1$, $k$-{\pats} is a variant of {\pats} that restricts input patterns to those with at most $k$ colors. 
	A computer-assisted proof has been recently proposed for 2-{\pats} by Kari et al.~[arXiv:1404.0967 (2014)]. 
	In contrast, the best known manually-checkable proof is for the \NP-hardness of 29-{\pats} by Johnsen, Kao, and Seki~[ISAAC 2013, LNCS 8283, pp.~699-710]. 
	We propose a manually-checkable proof for the \NP-hardness of 11-\pats. 
\end{abstract}

	\section{Introduction}

Tile self-assembly is an algorithmically rich model of ``programmable crystal growth.'' 
Well-designed molecules (square-like ``tiles'') with specific binding sites can deterministically form a single target shape even subject to the chaotic nature of molecules floating in a well-mixed chemical soup. 
Such tiles were experimentally implemented as DNA double-crossover molecules in 1998 \cite{WiLiWeSe1998}. 

\begin{figure}[tb]
\begin{center}
\begin{minipage}{0.4\linewidth}
\includegraphics[scale=0.45]{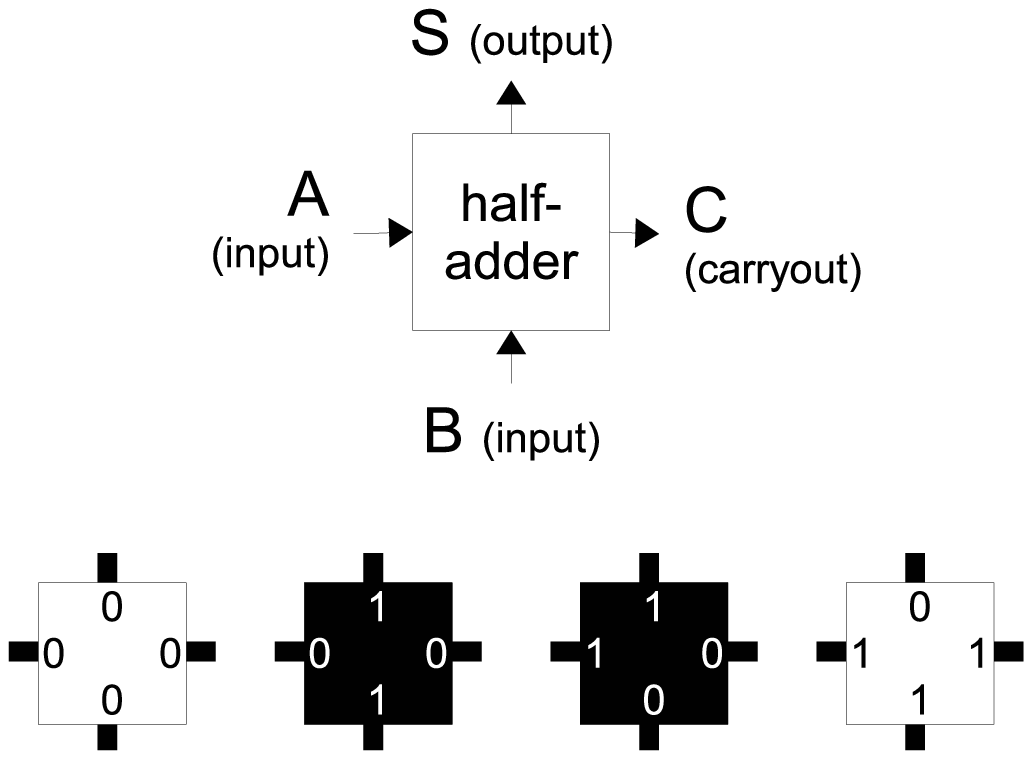}
\end{minipage}
\begin{minipage}{0.05\linewidth}
\hspace*{5mm}
\end{minipage}
\begin{minipage}{0.5\linewidth}
\includegraphics[scale=0.3]{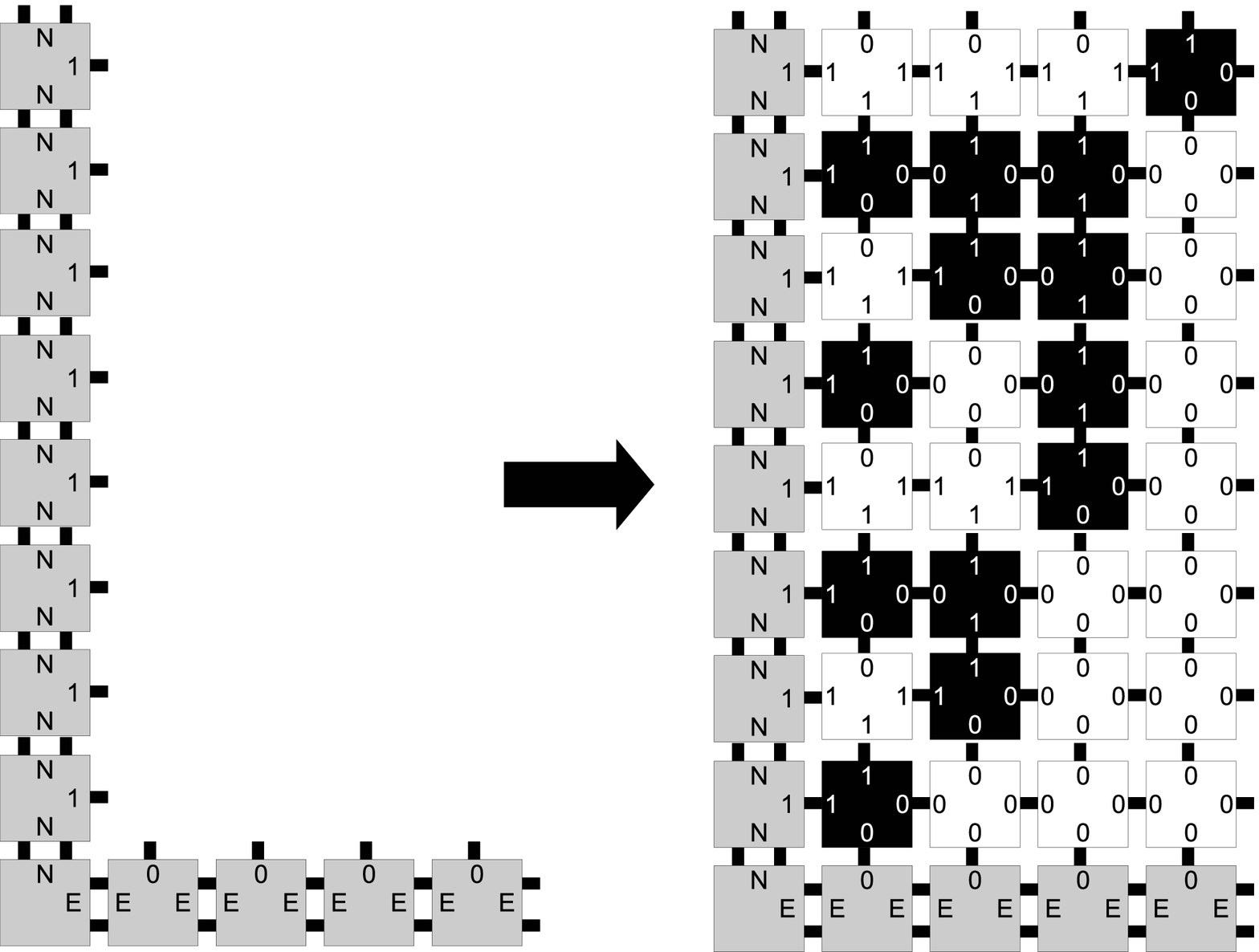}
\end{minipage}
\end{center}
\caption{
	(Left) Four tile types implement together the half-adder with two inputs A, B from the west and south, the output S to the north, and the carryout C to the east. 
	(Right) Copies of the ``half-adder'' tile types turn the L-shape seed into the binary counter pattern. 
}
\label{fig:SA_counter}
\end{figure}

Shape-building is one primary goal of self-assembly; pattern-painting is another. 
Based on the abstract Tile Assembly Model (aTAM) introduced by Winfree \cite{Winfree_PhDthesis}, Ma and Lombardi have first shed light on the pattern assembly \cite{MaLombardi2008,MaLombardi2009}. 
For the theory and practice of (color) pattern\footnote{``Pattern'' is a quite versatile term. In this paper, by pattern, we always mean a color pattern.} assembly, a simpler variant of TAM system (TAS) called the {\it rectilinear TAS} (RTAS) was proposed. 
As exemplified in Figure~\ref{fig:SA_counter}, an RTAS is provided with an L-shape seed (scaffold) as well as a finite number of tile types (the RTAS in the figure has four tile types: two white's and two black's) and their copies (i.e., tiles) attach to the seed and assemble a pattern (a binary counter pattern in the figure). 
The problem of {\em patterned self-assembly tile set synthesis} ({\pats}) aims at minimizing the number of tile types necessary for an RTAS to uniquely assemble a given rectangular pattern. 
An exhaustive partition-search algorithm as well as a randomized search algorithm \cite{GoLeCzOr2014} have been proposed for this problem. 

It is not until the number of colors included in the pattern is bounded by some constant $k \ge 1$ that {\pats} gets practically meaningful, as summarized in DNA 18\footnote{The 18th International Conference on DNA Computing and Molecular Programming.} as: ``{\it any given logic circuit can be formulated as a colored rectangular pattern with tiles, using only a constant number of colors}.'' 
We call this variant the {\em $k$-{\pats}}. 
The first result about $k$-{\pats} is the recent proof of the \NP-hardness of 60-{\pats} by Seki \cite{Seki2013} (2-{\pats} was claimed \NP-hard in \cite{MaLombardi2009}, but the proof was incorrect). 
Johnsen, Kao, and Seki strengthened the result up to the \NP-hardness of 29-{\pats} with $47/46 \approx 1.022$ being an approximation ratio unachievable in polynomial time, unless ${\bf P} = \NP$ \cite{JohnsenKaoSeki2013}. 

Kari et al. have recently proposed a computer-assisted proof for the \NP-hardness of 2-{\pats} \cite{KaKoMePaSe2014}. 
As a corollary of their proof, the approximation ratio $14/13 \approx 1.077$ is proven polynomial-time unachievable. 
Computer-assisted proofs are widely accepted these days, producing a number of results of practical value (see, e.g., \cite{KonevLisitsa2014, Marchal2011}). 
The proof for 2-{\pats} has been just verified in a different environment (computer architecture, programming language, etc.) from the first verification, and hence, it is very likely to be correct. 
Although the total computing time is almost 1-year, their programs are so massively parallelized that the actual verification takes just several days. 
This should be sufficient and shifts the practical interest onto the study of approximation algorithms. 

The aim of this paper is, nevertheless, to propose a manually-checkable proof of the \NP-hardness of 11-{\pats}. 
Beyond the aesthetic concerns about computer-assisted proofs (see quotations from Paul Erd\"{o}s in \cite{Hoffman1998}), manually-checkable proofs help us to obtain profound insights and understanding of the problem. 
This is a compilation of a series of works on the \NP-hardness of {\pats} \cite{CzeizlerPopa2013,JohnsenKaoSeki2013,KariKopeckiSeki2014,Seki2013}. 

\begin{theorem}\label{thm:11PATS_NPhard}
	11-{\pats} is \NP-hard. 
\end{theorem}

	\section{Rectilinear TAS and constant-colored {\pats}}

A {\it (rectangular) pattern $P$} (of width $w$ and height $h$) is a function from the rectangular domain $\{(x, y) \mid x \in \{0, 1, \ldots, w-1\}, y \in \{0, 1, \ldots, h-1\}\}$ to $\mathbb{N}$ (the set of color indices, or color codes). 
We denote the codomain of this pattern by ${\rm color}(P)$, that is, any color in ${\rm color}(P)$ appears at least once on $P$. 
We say that $P$ is {\it $k$-colored} if $|{\rm color}(P)| \le k$.

The self-assembly of binary counter (Figure~\ref{fig:SA_counter}) illustrates how a rectilinear TAS works. 
Let us first introduce necessary notation about the rectilinear TAS. 
A {\it tile type} is a square of some color whose four sides are {\it labeled}. 
Being assumed not to be rotatable, a tile type is identified by its color and four labels read in the counter-clockwise order starting at north (\north); for instance, the second black tile type in Figure~\ref{fig:SA_counter} (Left) is (1, 1, 0, 0, black). 
Given a tile type $t$ and a direction $d \in \{\north, \west, \south, \east\}$, $t(d)$ denotes the label at the side $d$. 
A {\it rectilinear TAS} (RTAS, in short) is a pair $\mathcal{T} = (T, \sigma_L)$ of a set $T$ of tile types and an L-shape seed $\sigma_L$ of width $w$ and height $h$ for some $w, h \ge 1$. 
As shown in Figure~\ref{fig:SA_counter}, the L-shape seed $\sigma_L$ is an assembly of tiles not included in $T$ so that its $x$-axis is provided with north labels and its $y$-axis is provided with east labels. 
Its domain is assumed to be $\{(0, 0)\} \cup \{(x, 0) \mid 1 \le x \le w\} \cup \{(0, y) \mid 1 \le y \le h\}$. 
The RTAS assumes an infinite supply of copies of tile types in $T$, each copy being referred to as a {\it tile}. 
Using the copies, it tiles the domain $\{(x, y) \mid 1 \le x \le w, 1 \le y \le h\}$ delimited by the seed, which is delimited by the L-shape seed, according to the following rule: 

\vspace*{3mm}

\begin{description}
\item[RTAS's tiling rule:] A tile can attach at a position $(x, y)$ if and only if its west label matches the east label of the tile on $(x-1, y)$ and its south label matches the north label of the tile on $(x, y-1)$. 
\end{description}

\vspace*{3mm}

\noindent
This rule suggests that a position does not become attachable until its west and south neighbor positions are tiled. 
At the initial time point, therefore, the sole attachable position is (1, 1). 
See the L-shape seed in Figure~\ref{fig:SA_counter} (Right); a tile of type (1, 1, 0, 0, black) can attach at (1, 1), while no tile of the other three types can attach, due to label-mismatching. 
The attachment makes the two positions (1, 2) and (2,~1) attachable. 
In this manner, the tiling proceeds from south-west to north-east {\it rectilinearly} until no attachable position is left. 
Since tile types are colored, if every position in the delimited domain has been tiled on the attachment termination, then the tiling shows a rectangular pattern and we consider it as an output of the RTAS and call it a {\it terminal pattern}. 
The $5 \times 9$ binary counter pattern in Figure~\ref{fig:SA_counter} is terminal. 
When an RTAS admits a unique terminal pattern $P$, we say that it {\it uniquely self-assembles the pattern $P$}. 

In this binary counter example, each attachable position admits a {\it unique} tile type whose copy (tile) can attach there, and we call this property directedness of RTAS. 
Formally, an RTAS $(T, \sigma_L)$ is {\it directed} if for any distinct $t_1, t_2 \in T$, either $t_1(\west) \neq t_2(\west)$ or $t_1(\south) \neq t_2(\south)$ holds (the directedness of RTAS was originally defined in a different but equivalent way). 
For technical convenience, we also say that such a tile type set $T$ is {\it directed}. 
It should be now clear that a directed RTAS uniquely self-assembles a pattern as long as it can tile the plain delimited by its seed. 

The {\it pattern self-assembly tile set synthesis} (\pats), proposed by Ma and Lombardi \cite{MaLombardi2008}, aims at computing the minimum size directed\footnote{Unlike the original form, the solution to {\pats} is required to be directed here, but it does not change the problem as the minimum RTAS is always directed \cite{GoLeCzOr2014}.} RTAS that uniquely self-assembles a given rectangular color pattern. 
The size of an RTAS $(T, \sigma_L)$ is measured solely by the cardinality of $T$, and is independent of the seed. 
By restricting the number of colors allowed to draw input patterns, a practically-meaningful subproblem of {\pats} is formulated as follows. 

\begin{definition}[\cite{Seki2013}]
	{\sc $k$-colored Pats} ($k$-{\pats}) \\
	\begin{tabular}{ll}
	{\sc Given}: & a $k$-colored pattern $P$ \\
	{\sc Find}: & a smallest directed RTAS that uniquely self-assembles $P$
	\end{tabular}
\end{definition}

	\section{Proof of Theorem~\ref{thm:11PATS_NPhard}}

Let us propose a polynomial-time reduction from monotone $\oneinthreeSAT$ to 11-{\pats} in the rest of this paper. 
$\oneinthreeSAT$ is a variant of $3\sat$ introduced by Schaefer \cite{Schaefer1978}. 
Its input is the same as the input of $3\sat$, while its decision is yes if and only if there exists an assignment that makes {\it exactly one} (compare to ``at least one'' in 3{\sat}) of the three literals in each clause true. 
It is \NP-hard \cite{Schaefer1978}, and it remains \NP-hard even under the restriction that no literal is negated; this restricted problem is called the {\it monotone $\oneinthreeSAT$}. 
An instance of monotone $\oneinthreeSAT$ is a conjunctive formula of clauses each of which consists of exactly three {\it positive} literals, i.e., variables. 

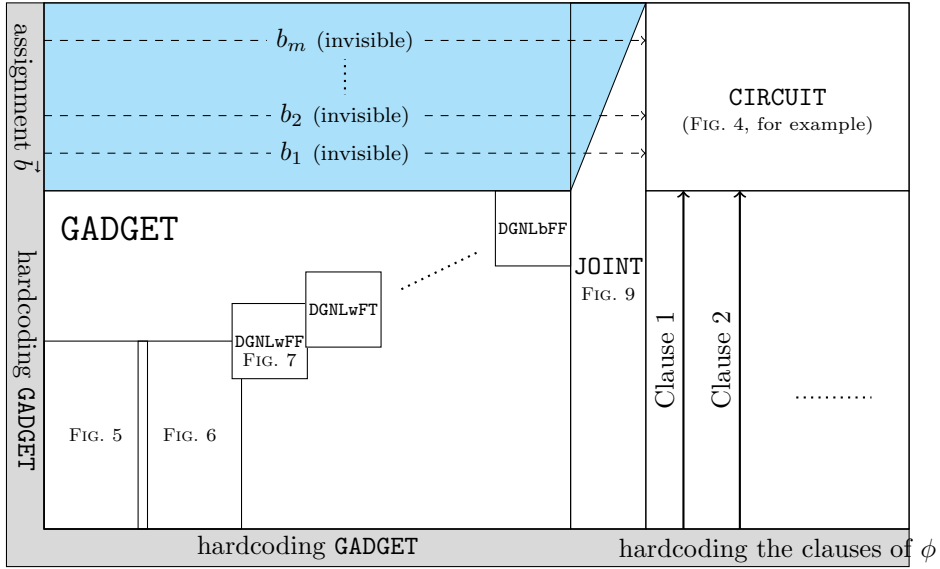
\begin{figure}[tb]
\begin{center}
\begin{tikzpicture}

\draw [fill=gray!30] (0, 0) -- (0, 7.5) -- (0.5, 7.5) -- (0.5, 0.5) -- (5.5, 0.5) -- (12, 0.5) -- (12, 0) -- cycle; 
\node () at (4, 0.25) {hardcoding {\tt GADGET}}; 
\node () at (0.25, 2.75) {\rotatebox{270}{hardcoding {\tt GADGET}}};
\node () at (0.25, 6.25) {\rotatebox{270}{assignment $\vec{b}$}};

\draw [fill=cyan!30] (0.5, 5) -- (7.5, 5) -- (8.5, 7.5) -- (0.5, 7.5) -- cycle; 
\draw (0.5, 0.5) rectangle (7.5, 5); \node at (1.5, 4.5) {\Large {\tt GADGET}};  
\draw (8.5, 5) rectangle node {{\tt CIRCUIT}} (12, 7.5); \node at (10.25, 5.875) { {\scriptsize ({\sc Fig.}~\ref{fig:circuit}, for example)}};  
\draw (0.5, 0.5) rectangle (12, 7.5); 
\draw (8.5, 0.5) -- node[below] {hardcoding the clauses of $\phi$} (12, 0.5) -- (12, 5) -- (8.5, 5) -- cycle; 
\draw (7.5, 0.5) rectangle node {\tt JOINT} (8.5, 7.5); \node at (8, 3.625) {\scriptsize {\sc Fig.}~\ref{fig:joint}}; 

\draw [->,dashed] (0.5, 7) -- node [fill=cyan!30] {$b_m$ {\footnotesize (invisible)}} (8.5, 7); 
\draw [dotted,thick] (4.5, 6.25) -- (4.5, 6.75);
\draw [->,dashed] (0.5, 6) -- node [fill=cyan!30] {$b_2$ {\footnotesize (invisible)}} (8.5, 6); 
\draw [->,dashed] (0.5, 5.5) -- node [fill=cyan!30] {$b_1$ {\footnotesize (invisible)}} (8.5, 5.5); 

\draw [->,thick] (9, 0.5) -- node [left] {\rotatebox{90}{Clause 1}} (9, 5);
\draw [->,thick] (9.75, 0.5) -- node [left] {\rotatebox{90}{Clause 2}} (9.75, 5);
\draw [dotted,thick] (10.5, 2.25) -- (11.5, 2.25);

\draw (0.5, 0.5) rectangle node {\scriptsize {\sc Fig.}~\ref{fig:gadget1}} (1.875, 3); 
\draw (1.75, 0.5) rectangle node {\scriptsize {\sc Fig.}~\ref{fig:gadget2}} (3.125, 3); 

\draw[fill=white] (3, 2.5) rectangle node {\scriptsize {\tt DGNLwFF}} (4, 3.5); \node at (3.5, 2.75) {\scriptsize {\sc Fig.}~\ref{fig:gadget_DGNLwFF}};
\draw[fill=white] (3.98, 2.92) rectangle node {\scriptsize {\tt DGNLwFT}} (4.98, 3.92); 
\draw (6.5, 4) rectangle node {\scriptsize {\tt DGNLbFF}} (7.5, 5); 

\draw[dotted, thick] (5.25, 3.69) -- (6.3, 4.2);  

\end{tikzpicture}
\end{center}
\caption{
	A blueprint of the pattern $P(\phi)$, to which a given monotone $\oneinthreeSAT$ instance $\phi$ is reduced.
	The Boolean-value assignment $\vec{b} = (b_1, b_2, \ldots, b_m)$ to $v_1, v_2, \ldots, v_m$ is invisible in the sense that the pattern $P(\phi)$ gives no information about it. 
	In contrast, the clauses are colorcoded on the pattern. 
}
\label{fig:blueprint}
\end{figure}

\begin{figure}[tb]
\begin{center}
\includegraphics{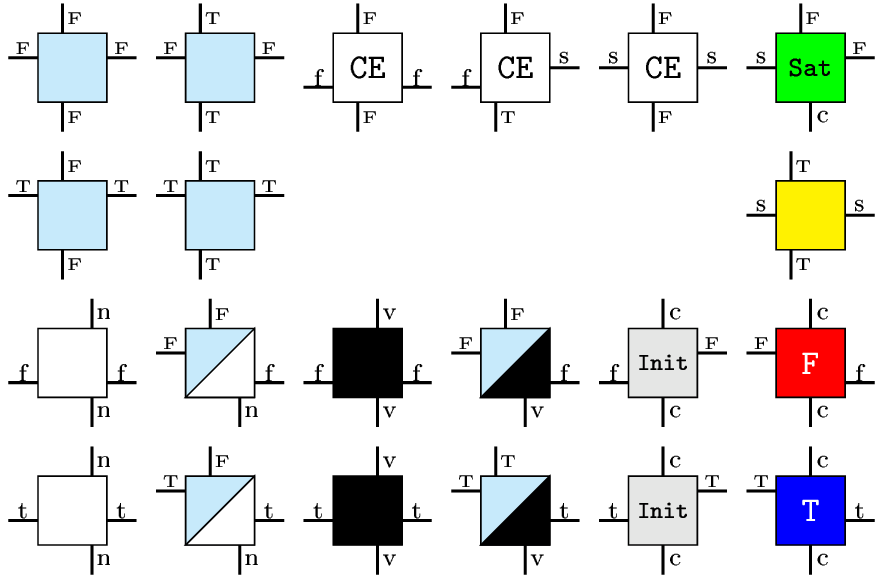}
\end{center}
\caption{Set $T_{\rm eval}$ of 21 tile types of 11 colors: cyan (4), {\tt CE} (3), white (2), black (2), {\tt DGNL}-white (2), {\tt DGNL}-black (2), {\tt Init} (2), {\tt Sat} (1), yellow (1), red (1), and blue (1), where the numbers in parentheses indicate how many tile types in $T_{\rm val}$ are drawn with corresponding colors.}
\label{fig:tileset}
\end{figure}

The set $T_{\rm eval}$ of 21 tile types, presented in Figure~\ref{fig:tileset}, is essential in our reduction. 
It is designed in such a way that, starting from an L-shape seed encoding a given monotone $\oneinthreeSAT$ instance $\phi$ over $m$ variables $v_1, v_2, \ldots, v_m$ and a Boolean-value assignment $\vec{b} = (b_1, b_2, \ldots, b_m)$ in a predetermined format on its glues, a directed RTAS with this tile type set evaluates $\phi$ according to $\vec{b}$ without revealing even a hint of $\vec{b}$ in the resulting pattern. 
We will explain this evaluation in detail in Section~\ref{subsec:circuit}. 

Our reduction converts a given instance $\phi$ of monotone $\oneinthreeSAT$ in an 11-colored rectangular pattern $P(\phi)$ consisting of primary and secondary subpatterns, as blueprinted in Figure~\ref{fig:blueprint}. 
The primary subpattern {\tt CIRCUIT} is a snapshot for $\phi$ to be thus validated (evaluated to be true) by tiles in $T_{\rm eval}$ according to some satisfying assignment $\vec{b}$. 
Needless to say, unless $\phi$ is satisfiable, the assignment is imaginary. 
The secondary subpattern {\tt GADGET} plays a critical auxiliary role in the reduction due to its following property: 

\begin{property}\label{propty:isomorphic}
	If a directed RTAS $(T, \sigma_L)$ with some set $T$ of at most 21 tile types uniquely self-assembles a pattern including {\tt GADGET}, then $T$ must be isomorphic to $T_{\rm eval}$ (modulo glue renaming). 
	Therefore, no set of strictly less than 21 tile types can be employed to uniquely self-assemble the pattern. 
\end{property}

{\tt GADGET} being included in the reduced pattern $P(\phi)$, Property~\ref{propty:isomorphic} forces a directed RTAS to employ $T_{\rm eval}$ in order to uniquely self-assemble $P(\phi)$, unless 22 or more tile types are available. 
Note that tiles in $T_{\rm eval}$ require an assignment satisfying $\phi$ to assemble the primary subpattern {\tt CIRCUIT} of $P(\phi)$, 
Consequently, $\phi$ is satisfiable if and only if $P(\phi)$ is uniquely self-assembled by a directed RTAS with at most 21 tile types. 

Having described informally how the reduction works. we will now explain it in detail in the rest of this paper. 

	\subsection{{\tt CIRCUIT}: validation of monotone $\oneinthreeSAT$}
	\label{subsec:circuit}

Using an example should be the easiest way to understand how, using tiles in $T_{\rm eval}$, a directed RTAS evaluates a monotone $\oneinthreeSAT$ instance according to a given assignment and what pattern will emerge as a result when the assignment satisfies the instance. 
Essentially, this pattern is {\tt CIRCUIT}. 

\begin{figure}[tb]
\begin{center}

\scalebox{0.57}{
\beginpgfgraphicnamed{circuit}
\begin{tikzpicture}
	\input{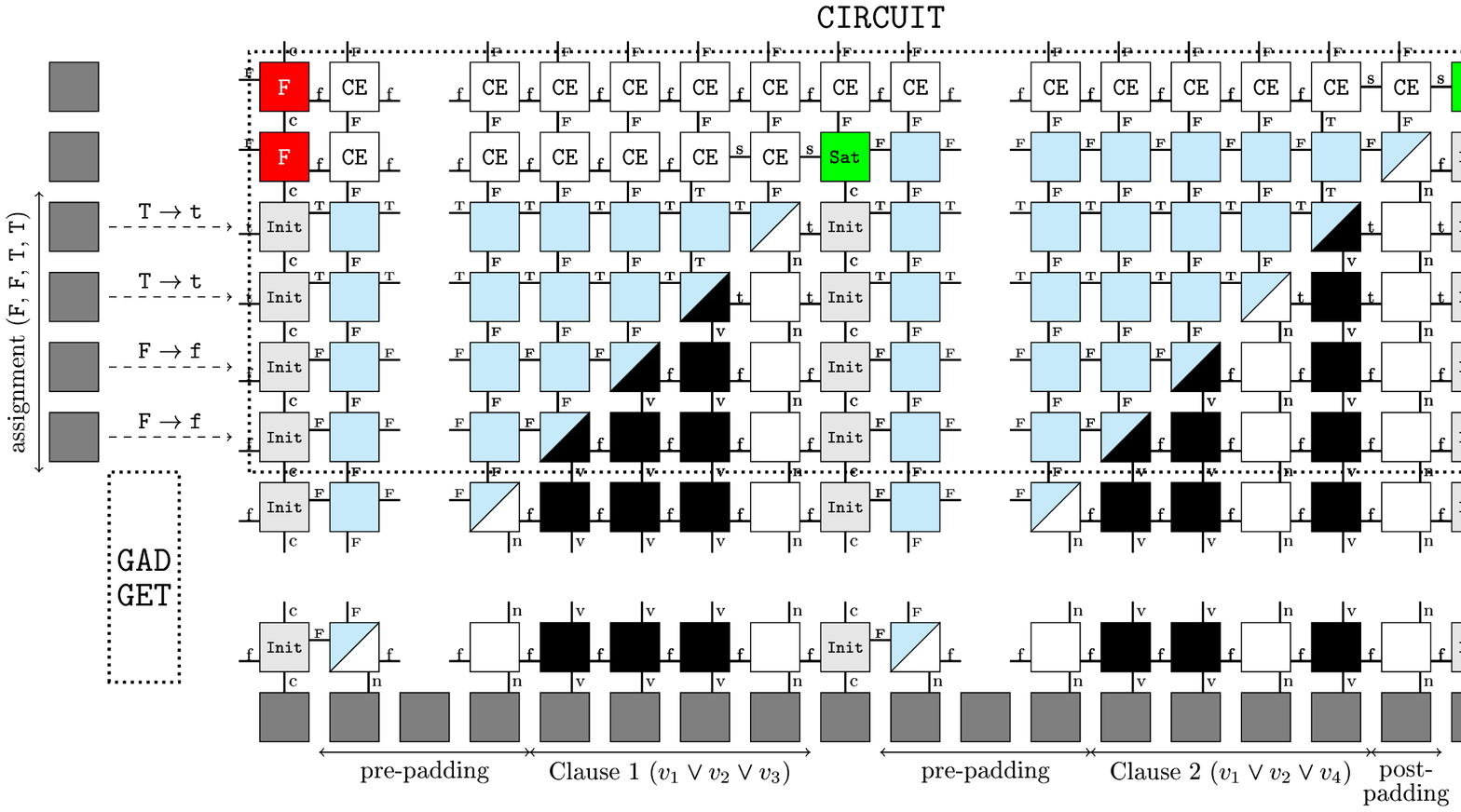}

	\foreach \x in {1, 2, ..., 18} {\node[tile,fill=gray] at (\x, 0) {};}
	\foreach \y in {4, 5, ..., 9} {\node[tile,fill=gray] at (-2, \y) {};}

	\draw[dashed, ->] (-1.5, 7) -- node[above] {${\tt T} \to {\tt t}$} (0.25,7); 
	\draw[dashed, ->] (-1.5, 6) -- node[above] {${\tt T} \to {\tt t}$} (0.25,6); 
	\draw[dashed, ->] (-1.5, 5) -- node[above] {${\tt F} \to {\tt f}$} (0.25,5); 
	\draw[dashed, ->] (-1.5, 4) -- node[above] {${\tt F} \to {\tt f}$} (0.25,4); 

	\draw[sloped,->] (-2.5, 3.5) -- node[above] {assignment ({\tt F}, {\tt F}, {\tt T}, {\tt T})} (-2.5, 7.5); \draw[->] (-2.5, 7.5) -- (-2.5, 3.5); 

	\draw[dotted, very thick] (-1.5, 0.5) rectangle (-0.5, 3.5); 
	\node at (-1, 2.25) {\Large {\tt GAD}}; 
	\node at (-1, 1.75) {\Large {\tt GET}}; 

	\draw[dotted, very thick] (0.5, 3.5) rectangle (18.5, 9.5); \node at (9.5, 10) {\Large {\tt CIRCUIT}}; 

	\draw[sloped,->] (1.5, -0.5) -- node[below] {pre-padding} (4.5, -0.5); \draw[->] (4.5, -0.5) -- (1.5, -0.5); 
	\draw[sloped,->] (4.5, -0.5) -- node[below] {Clause 1 $(v_1 \vee v_2 \vee v_3)$} (8.5, -0.5); \draw[->] (8.5, -0.5) -- (4.5, -0.5); 
	\draw[sloped,->] (9.5, -0.5) -- node[below] {pre-padding} (12.5, -0.5); \draw[->] (12.5, -0.5) -- (9.5, -0.5);  
	\draw[sloped,->] (12.5, -0.5) -- node[below] {Clause 2 $(v_1 \vee v_2 \vee v_4)$} (16.5, -0.5); \draw[->] (16.5, -0.5) -- (12.5, -0.5); 
	\draw[sloped,->] (16.5, -0.5) -- node[below] {post-} (17.5, -0.5); \draw[->] (17.5, -0.5) -- (16.5, -0.5); 
	\node at (17, -1.1) {padding};
\end{tikzpicture}
\endpgfgraphicnamed
}

\end{center}
\caption{
	Starting from the L-shape seed, indicated by gray tiles, that encodes the instance $\phi = (v_1 \vee v_2 \vee v_3) \wedge (v_1 \vee v_2 \vee v_4)$ and an assignment $\vec{b}=(\false, \false, \true, \true)$, a directed RTAS evaluates $\phi$ according to $\vec{b}$ using tiles in $T_{\rm eval}$.
	The assembly results in the subpattern {\tt CIRCUIT} to the northeast of {\tt GADGET} on $P(\phi)$. 
}
\label{fig:circuit}
\end{figure}

Consider a formula $\phi = (v_1 \vee v_2 \vee v_3) \wedge (v_1 \vee v_2 \vee v_4)$ and an assignment $\vec{b} = (\false, \false, \true, \true)$, which satisfies $\phi$ in the $\oneinthreeSAT$ sense\footnote{In contrast, $(\true, \false, \true, \false)$ does not satisfy $\phi$ in the $\oneinthreeSAT$ sense because it satisfies more than one literal of the first clause.}. 
See Figure~\ref{fig:circuit} for the evaluation of $\phi$ according to $\vec{b}$ by the RTAS. 

The L-shape seed is the interface to input $\phi$ and $\vec{b}$ into the RTAS. 
Clauses of $\phi$ are written on the seed's $x$-axis as a sequence of glues {\tt v} (variable in clause), {\tt n} (variable not in clause\footnote{{\tt n} does not mean a negated variable. Recall, monotone implies variables never appear negated in clauses}), and {\tt c}. 
The clauses $(v_1 \vee v_2 \vee v_3)$ and $(v_1 \vee v_2 \vee v_4)$ of $\phi$, for instance, are first converted into ${\tt vvvn}$ and ${\tt vvnv}$, respectively. 
We then pre-pad each of these encodings from the left by $h$ {\tt n} glues so that {\tt CIRCUIT} is to emerge at the height $h$. 
Later, $h$ will be set to the height of {\tt GADGET}. 
We finally post-pad them with incremental number of {\tt n} glues so that a clause is evaluated on the row just above those on which previous clauses were evaluated. 
Connecting them by {\tt c}'s results in ${\tt c} {\tt \underline{n}}^h {\tt vvvn \underline{n}^0 c} {\tt \underline{n}}^h {\tt vvnv \underline{n}^1 c}$, where {\tt n}'s for padding are underlined. 
This is the encoding of the clauses of $\phi$. 
The assignment $\vec{b}$ is written rather on the seed's $y$-axis as {\tt FFTT} (the assignment to the first variable $v_1$ is at the bottom). 
We post-pad it with as many {\tt F}'s as clauses of $\phi$ like {\tt FFTT}-${\tt F}^2$ for this example. 

Signals {\tt v} and {\tt n}, carrying information about the membership of variables in clauses, are propagated northward through black and white tiles (2 types each), respectively. 
The clauses become visible in this way. 
Cyan tiles (4 types) propagate signals ({\tt F}/{\tt T}) horizontally as well as vertically. 
The assignment is thus propagated horizontally over {\tt GADGET} by cyan tiles, and lower-cased ($\false/\true \to {\tt f}/{\tt t}$) when passing the joint between {\tt GADGET} and {\tt CIRCUIT} (see Figure~\ref{fig:joint}). 

At the crossover of these signals, variables are evaluated diagonally by {\tt DGNL}-black tiles (2 types); they reflect the signal from the west (assignment) to the north like a mirror. 
The three signals thus evaluated per clause are propagated to the north via cyan tiles and then {\tt CE} tiles (3 types) evaluate these signals. 
At an encounter with {\tt T} signal, {\tt CE} tiles change the evaluation from {\tt f} to {\tt s} (satisfied), and without another encounter with {\tt T} signal, {\tt CE} tiles propagate the evaluation to the east until it is validated by a {\tt Sat} tile at the top of {\tt Init} column, which initializes the assignment signals for the validation of the next clause. 
The post-padding enables clauses to be evaluated on different rows. 

See Figure~\ref{fig:circuit} for the emerging pattern {\tt CIRCUIT}. 
What has to be observed is the invisibility of the assignment $\vec{b}$ encoded in the seed on the pattern. 
The assignment can be retrieved only by examining its underlying assembly, and not by the colors of its pattern. 
In fact, from two L-shape seeds encoding different satisfying assignments in the above-mentioned format, tiles in $T_{\rm eval}$ assemble the same pattern {\tt CIRCUIT}. 
It might be also worthwhile to note that starting from the seed which encodes an unsatisfying assignment, the RTAS cannot complete any rectangular pattern due to the lack of the {\tt UNSAT} counterpart of {\tt SAT} tile type or the {\tt CE} tile type receiving {\tt s} from the west and {\tt T} from the south to handle a second true literal in $T_{\rm eval}$. 

{\tt CIRCUIT} involves just 9 colors: cyan, {\tt CE}, white, {\tt DGNL}-white, black, {\tt DGNL}-black, {\tt Init}, red (F), and {\tt Sat}. 
Yellow and blue (T) appear on the secondary subpattern {\tt GADGET} so that the whole pattern $P(\phi)$ is 11-colored. 

	\subsection{Secondary subpattern {\tt GADGET}}
	\label{subsec:gadget}

We have seen that if $\phi$ is satisfiable, then a directed RTAS can self-assemble {\tt CIRCUIT} using tiles in $T_{\rm eval}$. 
In Figures~\ref{fig:gadget1}-\ref{fig:joint}, we visualize how tiles in $T_{\rm eval}$ self-assemble other parts of $P(\phi)$ (see in Figure~\ref{fig:blueprint} how they are integrated into $P(\phi)$). 
These should be enough for us to be convinced that if $\phi$ is satisfiable, then a directed RTAS uniquely self-assembles the pattern $P(\phi)$. 

The converse implication is much harder to be proved: if a directed RTAS with at most 21 tile types uniquely self-assembles $P(\phi)$, then $\phi$ is satisfiable. 
This is primarily because of the huge number of possible tile type sets as well as possible seeds for the RTAS. 
The role of {\tt GADGET} is to make all tile type sets but $T_{\rm eval}$ useless (Property~\ref{propty:isomorphic}). 
That is, with at most 21 tile types available, the RTAS must employ $T_{\rm eval}$ to uniquely self-assemble $P(\phi)$. 
The RTAS still has the freedom of choice in its seed. 
However, at the top of the $y$-axis, the seed's glues must be of the form $({\tt F}/{\tt T})^m {\tt F}^k$ (see Figure~\ref{fig:joint}), where $m$ and $k$ refer to the number of variables and clauses in $\phi$, respectively. 
This is because the west glue of cyan tiles in $T_{\rm eval}$ is either {\tt F} or {\tt T} and that of the red(F) tile type is {\tt F}. 
The choice of the specific glue sequence for $({\tt F}/{\tt T})^m$ among all possible $2^m$ candidates corresponds to an assignment of false/true values to the $m$ variables of $\phi$. 
The above-mentioned invisibility of the assignment allows the RTAS to make this choice, but the chosen one must satisfy $\phi$ in order to assemble {\tt CIRCUIT} of $P(\phi)$ completely. 
Thus, $\phi$ is satisfiable. 
The proof of Theorem~\ref{thm:11PATS_NPhard} is completed in this way. 

\begin{figure}[tbp]
\begin{center}

\scalebox{0.60}{
\beginpgfgraphicnamed{gadget1}
\begin{tikzpicture}

\input{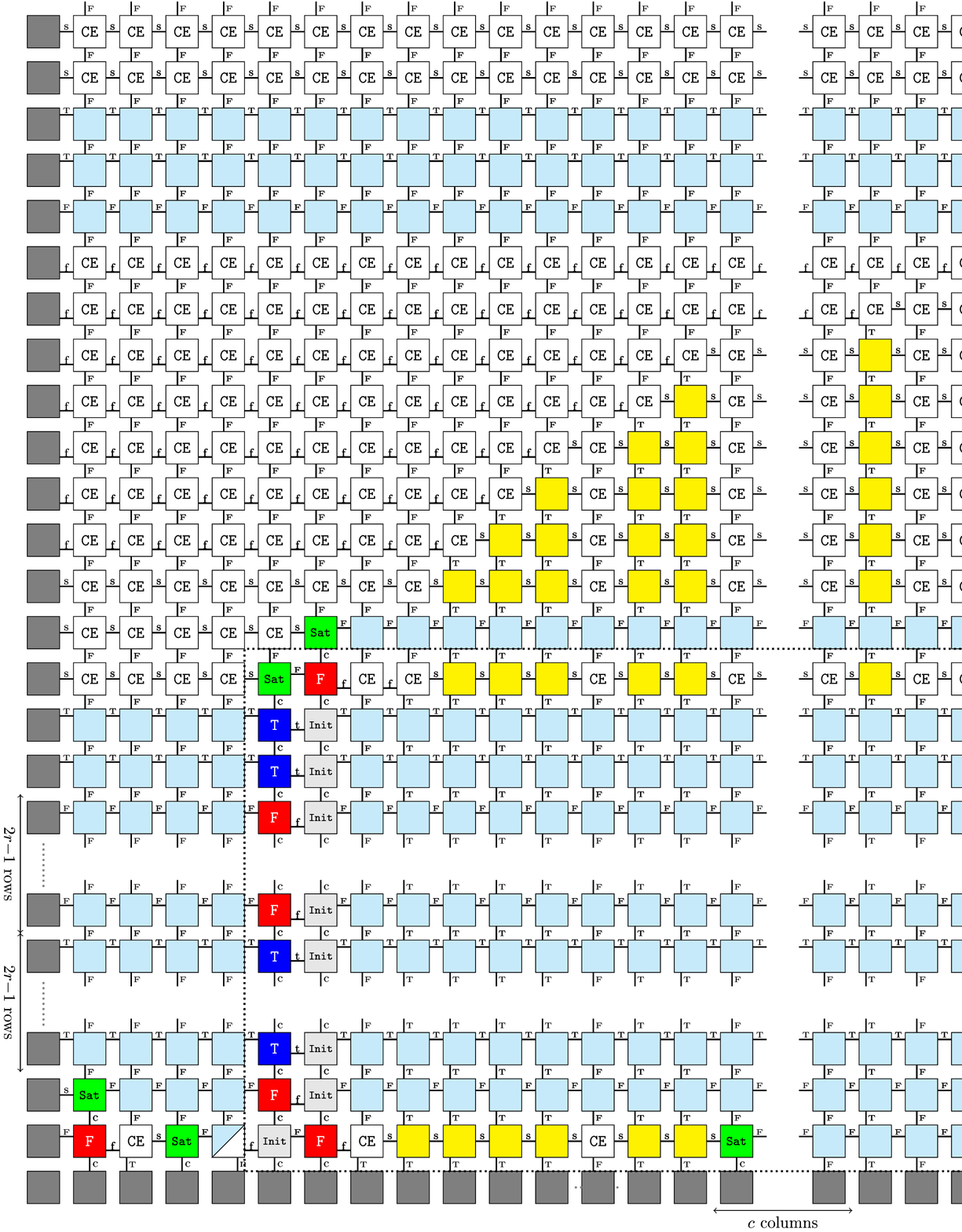}
\seed{15}{3};
\foreach \x in {17, 18, 19, 20, 21} {\node[tile, fill=gray] at (\x, 0) {};}
\node[tile,fill=gray] at (0, 5) {}; \node[tile,fill=gray] at (0, 6) {}; 
\foreach \y in {8, 9, ..., 25} {\node[tile,fill=gray] at (0, \y) {};}
\draw[dotted, very thick, draw=gray] (11.5, 0) -- (12.5, 0); 
\draw[dotted, very thick, draw=gray] (0, 3.5) -- (0, 4.5); 
\draw[dotted, very thick, draw=gray] (0, 6.5) -- (0, 7.5); 

\draw[dotted, very thick] (4.35,0.35) rectangle (21.65,11.65) node [right] {\Large {\tt LB4}};

\draw[->,sloped] (-0.5, 5.5) -- node[below] {$2r{-}1$ rows} (-0.5, 2.5); \draw[->] (-0.5, 2.5) -- (-0.5, 5.5); 
\draw[->,sloped] (-0.5, 8.5) -- node[below] {$2r{-}1$ rows} (-0.5, 5.5); \draw[->] (-0.5, 5.5) -- (-0.5, 8.5); 

\draw[->] (14.5, -0.5) -- node[below] {$c$ columns} (17.5, -0.5); \draw[->] (17.5, -0.5) -- (14.5, -0.5); 

\end{tikzpicture}
\endpgfgraphicnamed
}

\end{center}
\caption{
	The leftmost part of the secondary subpattern {\tt GADGET} of the reduced pattern $P(\phi)$.
	The constants $c$ and $r$, which are independent of $\phi$, are set large enough for the proof's sake. 
}
\label{fig:gadget1}
\end{figure}

\begin{figure}[tbp]
\begin{center}

\scalebox{0.55}{
\beginpgfgraphicnamed{gadget2}
\begin{tikzpicture}

\input{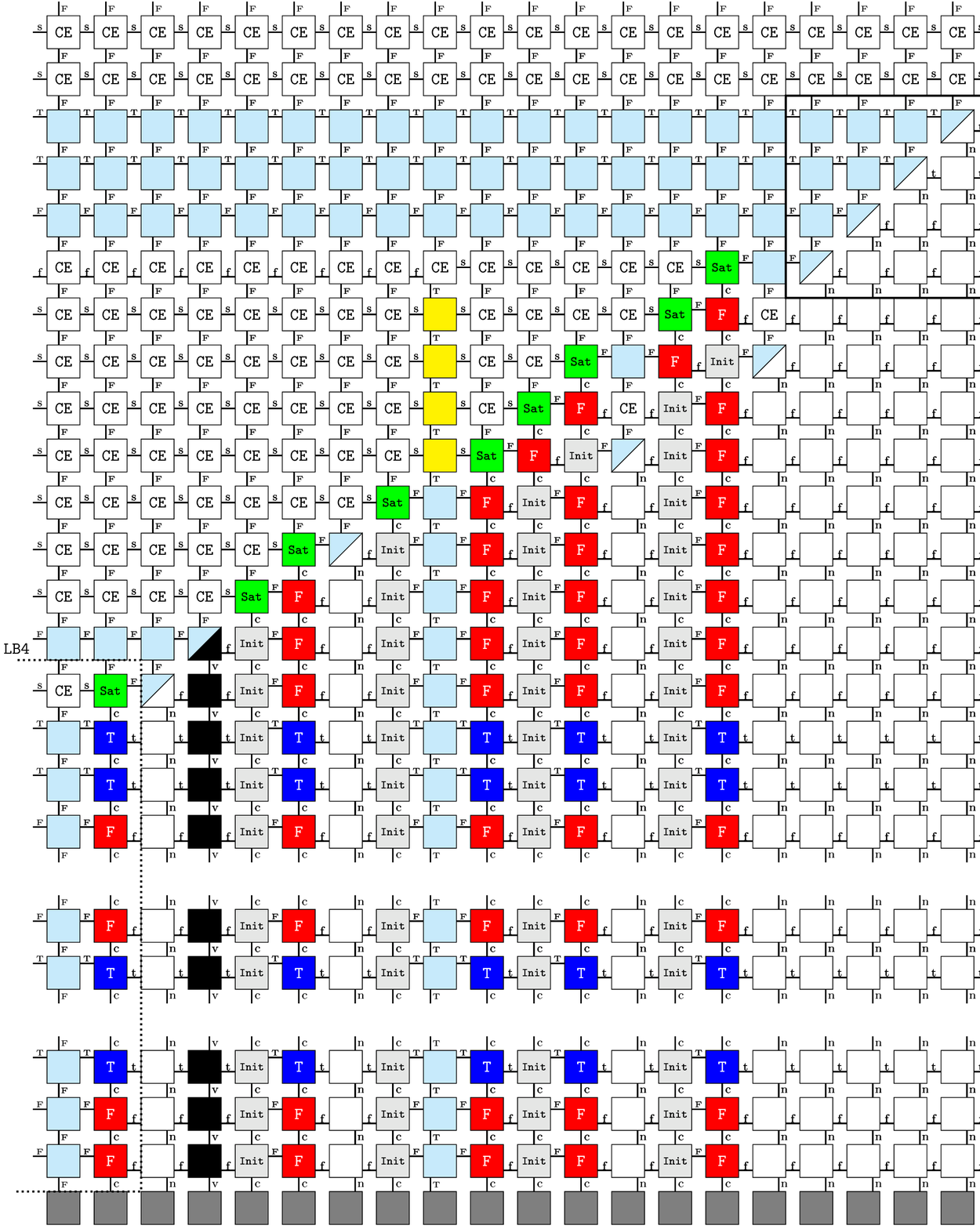}

\foreach \x in {1, 2, ..., 22} {\node[tile,fill=gray] at (\x, 0) {};}

\draw[->,sloped] (22.5, 5.5) -- node[below] {$2r{-}1$ rows} (22.5, 8.5); \draw[->] (22.5, 8.5) -- (22.5, 5.5); 
\draw[->,sloped] (22.5, 2.5) -- node[below] {$2r{-}1$ rows} (22.5, 5.5); \draw[->] (22.5, 5.5) -- (22.5, 2.5); 

\draw [very thick] (16.35, 19.35) rectangle (22.65, 23.65);  

\draw[dotted, very thick] (0,0.35) -- (2.65, 0.35) -- (2.65, 11.65) -- (0, 11.65) node[above] {{\tt LB4}}; 

\end{tikzpicture}
\endpgfgraphicnamed
}

\end{center}
\caption{
	The middle part of {\tt GADGET}.
	In order to clarify that this subpattern is located to the east of the one in Figure~\ref{fig:gadget1} on $P(\phi)$, this figure includes the easternmost two columns in Figure~\ref{fig:gadget1}.
	As for the framed subpattern to the northeast, see the legend of Figure~\ref{fig:gadget_DGNLwFF}. 
}
\label{fig:gadget2}
\end{figure}

\begin{figure}[tbp]
\begin{center}
\beginpgfgraphicnamed{gadget_DGNLwFF}
\scalebox{0.5}{
\begin{tikzpicture}

\input{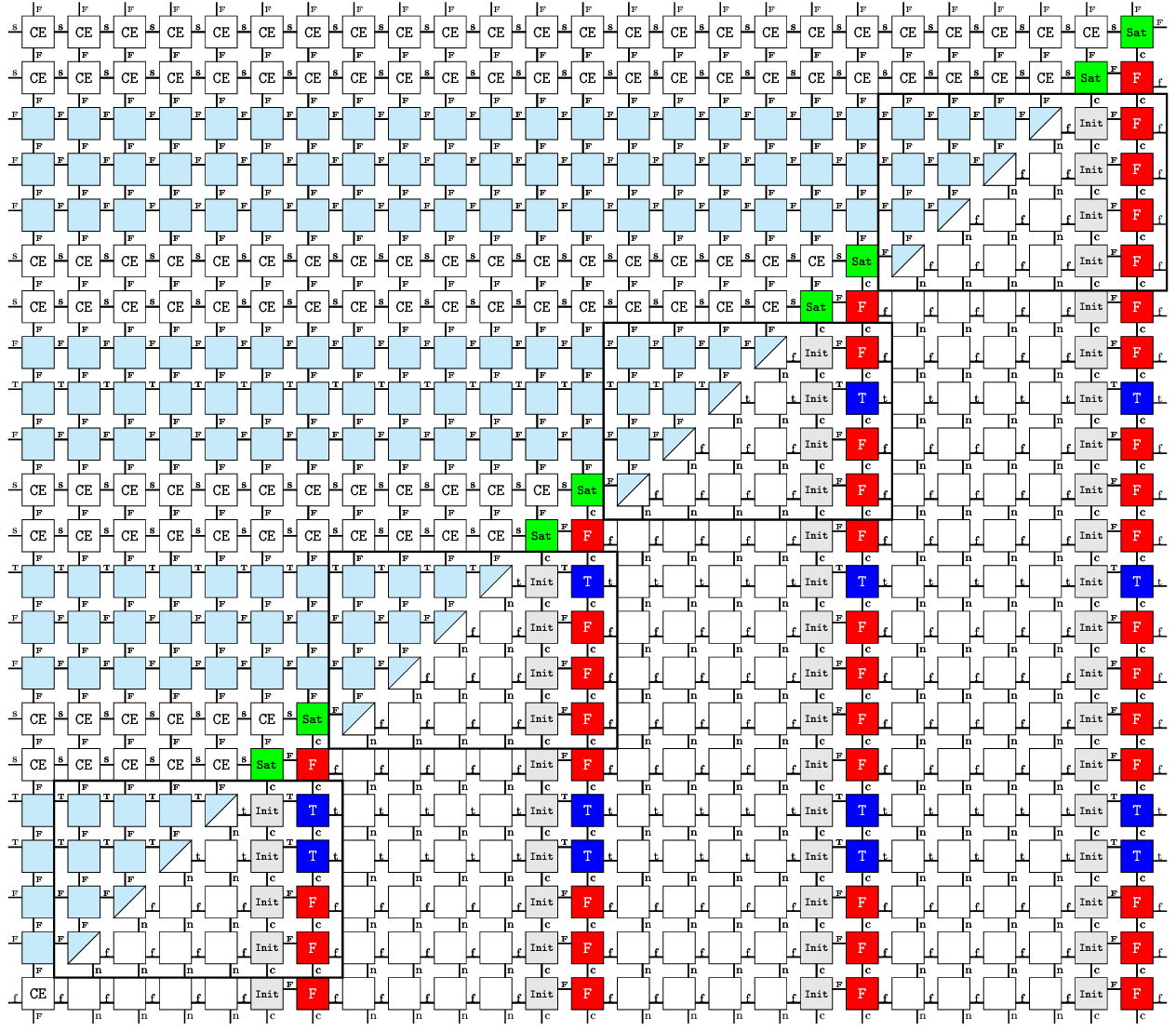}
\draw[very thick] (2.35, 1.35) rectangle (8.65, 5.65); 
\draw[very thick] (8.35, 6.35) rectangle (14.65, 10.65); 
\draw[very thick] (14.35, 11.35) rectangle (20.65, 15.65); 
\draw[very thick] (20.35, 16.35) rectangle (26.65, 20.65); 

\end{tikzpicture}
}
\endpgfgraphicnamed

\end{center}
\caption{
	The first one-eighth of the rightmost part of {\tt GADGET}.
	The four subpatterns framed are the instances of the template described in Figure~\ref{fig:lb2-DGNL} with two red (F) at the bottom. 
	The leftmost one of them is actually the one on the middle part shown in Figure~\ref{fig:gadget2}, and suggests that this is located to the northeast of the middle part. 
}
\label{fig:gadget_DGNLwFF}
\end{figure}

\begin{figure}[tbp]
\begin{center}
\beginpgfgraphicnamed{gadget_DGNLbTT}
\scalebox{0.5}{
\begin{tikzpicture}

\input{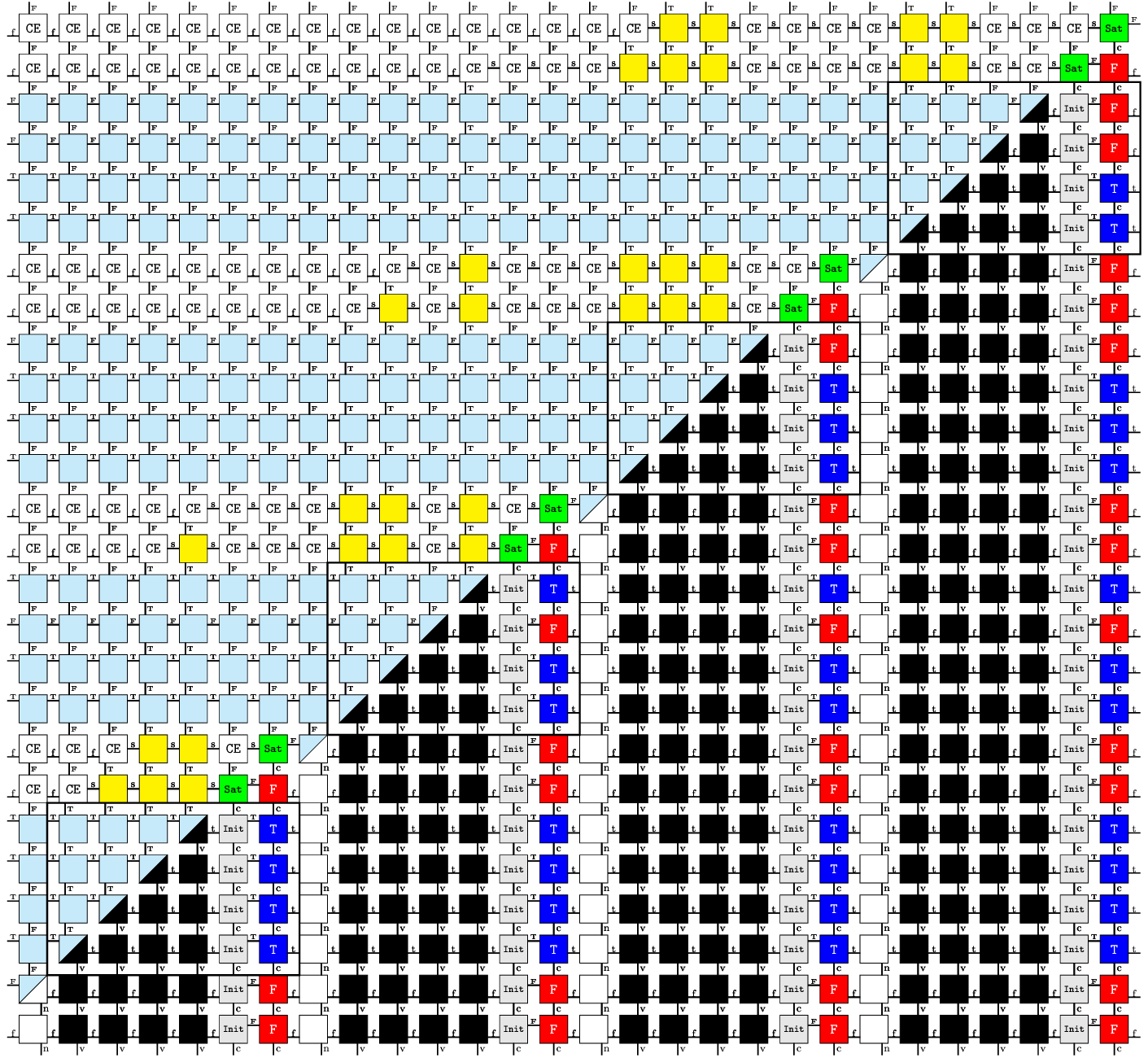}
\draw[very thick] (3.35, 2.35) rectangle (9.65, 6.65); 
\draw[very thick] (10.35, 8.35) rectangle (16.65, 12.65); 
\draw[very thick] (17.35, 14.35) rectangle (23.65, 18.65); 
\draw[very thick] (24.35, 20.35) rectangle (30.65, 24.65); 

\end{tikzpicture}
}
\endpgfgraphicnamed

\end{center}
\caption{
	Another one-eighth of the rightmost part of {\tt GADGET}.
	The four subpatterns framed are the instances of the black analogue of the template described in Figure~\ref{fig:lb2-DGNL} with two blue (T) at the bottom. 
}
\label{fig:gadget_DGNLbTT}
\end{figure}

\begin{figure}[tb]
\begin{center}

\scalebox{0.55}{
\beginpgfgraphicnamed{joint}
\begin{tikzpicture}
	\input{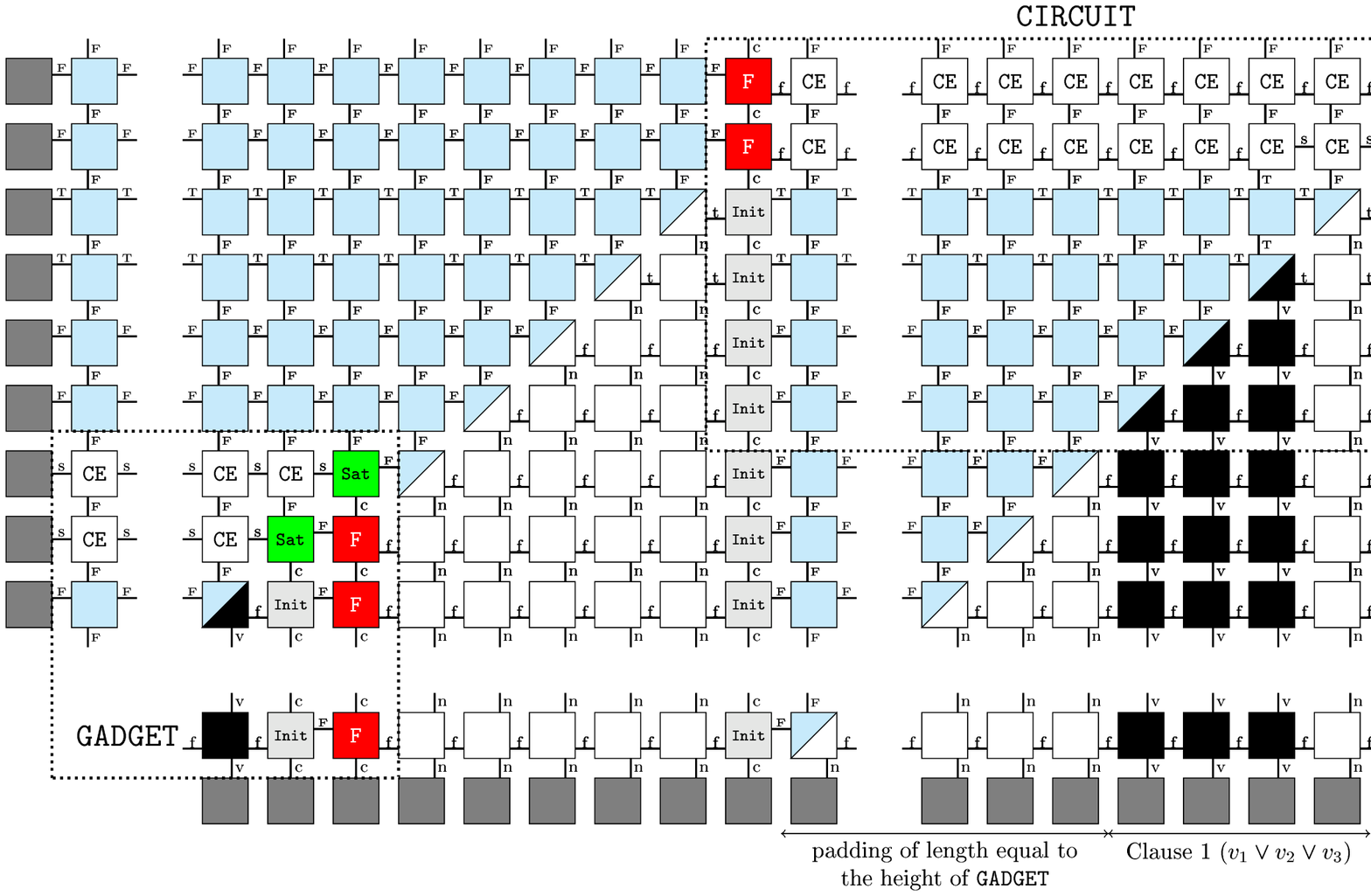}
	\foreach \x in {3, 4, ..., 12} {\node[tile,fill=gray] at (\x, 0) {};}
	\foreach \x in {14, 15, ..., 21} {\node[tile,fill=gray] at (\x, 0) {};}
	\foreach \y in {3, 4, ..., 11} {\node[tile,fill=gray] at (0, \y) {};}

	\draw[->] (11.5, -0.5) -- node[below] {padding of length equal to} (16.5, -0.5); \draw[->] (16.5, -0.5) -- (11.5, -0.5); \node at (14, -1.2) {the height of {\tt GADGET}}; 
	\draw[->] (16.5, -0.5) -- node[below] {Clause 1 $(v_1 \vee v_2 \vee v_3)$} (20.5, -0.5); \draw[->] (20.5, -0.5) -- (16.5, -0.5); 

	\draw[dotted,very thick] (0.35, 0.35) rectangle (5.65, 5.65); 
	\node at (1.5, 1) {\Large {\tt GADGET}}; 
	\draw[dotted,very thick] (21.65, 11.65) -- (10.35, 11.65) -- (10.35, 5.35) -- (21.65, 5.35); 
	\node at (16, 12) {\Large {\tt CIRCUIT}}; 
\end{tikzpicture}
\endpgfgraphicnamed
}
\end{center}
\caption{A joint between {\tt GADGET} and {\tt CIRCUIT}.}
\label{fig:joint}
\end{figure}

\begin{figure}[htb]
\begin{center}

\begin{minipage}{0.35\linewidth}

\begin{tikzpicture}

\draw[fill=cyan!20] (-0.35, -0.35) rectangle (2.45, 2.45); 
\draw[fill=white] (-0.35, -0.35) -- (2.45, 2.45) -- (2.45, -0.35) -- cycle;

\foreach \x in {0.35, 1.05, 1.75} {
	\draw (\x, -0.35) -- (\x, 2.45); 
}
\foreach \y in {0.35, 1.05, 1.75} {
	\draw (-0.35, \y) -- (2.45, \y); 
}

\foreach \y in {0, 0.7, 1.4, 2.1} {
	\node[graytile] at (2.8, \y) {};
	\node at (2.8, \y) {\scriptsize {\tt Init}};
	\node[tile, dashed, draw=purple, text=purple] at (3.5, \y) {\scriptsize {\tt T/F}}; 
}

\end{tikzpicture}

\end{minipage}
\begin{minipage}{0.05\linewidth}
\ \\
\end{minipage}
\begin{minipage}{0.55\linewidth}
\scalebox{0.9}{

\beginpgfgraphicnamed{template-instance}
\begin{tikzpicture}
\input{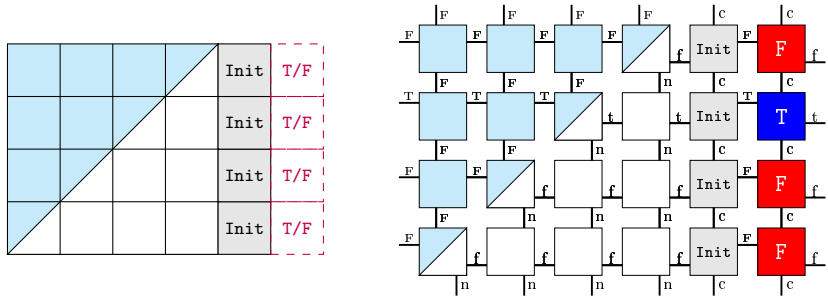}
\end{tikzpicture}
\endpgfgraphicnamed

}
\end{minipage}

\end{center}
\caption{
	{\rm (Left)} Template of 16 subpatterns of {\tt GADGET} that give lower bound 2 on the number of {\tt DGNL}-white tile types, where each of the four purple positions is either blue or red. 
	{\rm (Right)} The assembly of one of the 16 subpatterns by tiles in $T_{\rm eval}$. 
}
\label{fig:lb2-DGNL}
\end{figure}

Before verifying Property~\ref{propty:isomorphic} in Section~\ref{subsec:property_proof}, we should explain the constitution of {\tt GADGET} and how it is integrated, together with {\tt CIRCUIT}, into the pattern $P(\phi)$. 
{\tt GADGET} is composed of three parts: leftmost one including an important subpattern {\tt LB4} (Figure~\ref{fig:gadget1}), middle part (Figure~\ref{fig:gadget2}), and rightmost part. 
The subpattern {\tt LB4} is parameterized by two constants $c$ and $r$, which are set large enough for the sake of our proof of Lemma~\ref{lem:LB4} below (their actual values shall be specified at the beginning of the proof). 
It must be noted that these constants are independent of the size and clauses of $\phi$. 
As for the rightmost part, it is further split into eight parts due to its size; a one-eighth {\tt DGNLwFF} is sketched in Figure~\ref{fig:gadget_DGNLwFF} and another one-eighth {\tt DGNLbTT} is sketched in Figure~\ref{fig:gadget_DGNLbTT}. 
These parts contain the sixteen instances of a subpattern template shown in Figure~\ref{fig:lb2-DGNL} (Left) and their sixteen black analogues. 
The eight one-eighths are positioned at the northeastern corner of {\tt GADGET}; their order does not matter, but we choose {\tt wFF}--{\tt wFT}--{\tt wTF}--{\tt wTT}--{\tt bTT}--{\tt bTF}--{\tt bFT}--{\tt bFF}; here {\tt DGNL} were omitted. 

{\tt GADGET} is meticulously designed so that, being assembled from tiles in $T_{\rm eval}$, it exposes 
\begin{itemize}
\item	only {\tt F} glues to the north;
\item	only {\tt f}/{\tt t} glues to the east, except at the top where the glue is {\tt F}. 
\end{itemize}
The north {\tt F} glues enable cyan tiles to attach to their north and propagate the assignment above {\tt GADGET} toward {\tt CIRCUIT} invisibly. 
With $m$ {\tt n} glues on the $x$-axis of the seed,\footnote{Recall that $m$ is the number of variables involved in $\phi$.} the east {\tt f}/{\tt t} glues let white tiles assemble the foundation of {\tt JOINT} on which DGNL-white tiles attach diagonally in collaboration with cyan and white tiles and lower-case the assignment signals $(\false/\true \to {\tt f}/{\tt t})$ (see Figure~\ref{fig:joint}). 
{\tt CIRCUIT} and {\tt GADGET} are thus integrated into the pattern $P(\phi)$. 

	\subsection{Verification of Property~\ref{propty:isomorphic}}
	\label{subsec:property_proof}

The aim of this subsection is to verify Property~\ref{propty:isomorphic}, and hence, conclude the proof of Theorem~\ref{thm:11PATS_NPhard}. 
The verification is done through the following task: given 21 tile types which have not been colored or labelled yet, color and label them so that, using the resulting tile type set, a directed RTAS can uniquely self-assemble {\tt GADGET}. 

	\subsubsection{Coloring}

Let us handle coloring first; we will observe that the given 21 tile types must be colored as $T_{\rm eval}$ does: 4 cyan, 3 {\tt CE}, 2 white, {\tt DGNL}-white, black, {\tt DGNL}-black, {\tt Init} each, and 1 {\tt Sat}, yellow, red (F), and blue (T) each. 
In fact, we only have a room to choose colors of 10 of them because with each color, at least one tile type must be painted. 

We begin with the need for one more {\tt Init} tile type. 
For the sake of contradiction, suppose there were only one {\tt Init} tile type. 
See the rightmost column in Figure~\ref{fig:gadget2}. 
At its bottom, 2 red (F) and $2r{-}1$ blue (T) positions are found, and on top of them is one more red position (at the height $2r{+}2$). 
Since their western neighbors are all {\tt Init}, with only one {\tt Init} tile type, a directed RTAS would need to fill the blue positions with $2r{-}1$ tiles of pairwise distinct types in order to attach a red tile precisely at the height $2r{+}2$ (the hardcoded height). 
This would cost the RTAS an unaffordable $2r{-}2$ extra blue tile types (recall that $r$ was set large enough). 
Thus, we need to draw one uncolored tile type by {\tt Init} and 9 tile types remain uncolored. 

To their west is a white column (the third from the right). 
With only one white tile type, we find that the red position on top of the $2r{-}1$ blue positions must again be hardcoded from below through the {\tt Init} and red (F)/blue (T) columns. 
This is, however, unaffordable, provided $r$ is set sufficiently large. 
The same argument based rather on the fourth, fifth, and sixth leftmost columns in Figure~\ref{fig:gadget2} justifies the need of at least 2 black tile types. 
Among the 9 uncolored tile types, one has been drawn white and another has been drawn black. 
As a result, 7 tile types remain uncolored. 

Before painting them, let us present one lemma on {\tt Init}, white, and black tile types. 

\begin{lemma}\label{lem:exactly2}
	Let $col \in \{{\tt init}, white, black\}$. 
	If a directed RTAS with at most 21 tile types including exactly 2 tile types $t_1, t_2$ of color $col$ uniquely self-assembles a pattern including {\tt GADGET}, then $t_1(\west) \neq t_2(\west)$ and $t_1(\east) \neq t_2(\east)$, while $t_1(\south) = t_2(\south)$.
\end{lemma}
\begin{proof}
	We prove this lemma only for $col = {\tt Init}$. 
	We have already seen the need for $t_1(\east) \neq t_2(\east)$; otherwise hardcoding would be necessary in order to place the red tile at the specific height. 

	Suppose $t_1(\south)$ were different from $t_2(\south)$. 
	This distinctness forces the RTAS to assemble the second rightmost column in Figure~\ref{fig:gadget2} periodically either as $t_1t_2t_1t_2 \cdots$ or as $t_1t_2 t_2 \cdots$. 
	In any case, the column exposes a periodic sequence of east glues, and hence, the placement of the red tile at the specific height would require the unaffordable cost in hardcoding by blue tile types. 
	Therefore, $t_1(\south) = t_2(\south)$ must hold, and this implies $t_1(\west) \neq t_2(\west)$ in order for the RTAS to be directed. 
\end{proof}

Next, we focus on cyan tiles. 
As of now, just 1 tile type was drawn cyan. 
We will show that due to the subpattern {\tt LB4} in Figure~\ref{fig:gadget1}, designated by a dotted rectangle, we need either 3 more cyan tile types, or 2 more cyan tile types and 2 more tile types whose color is either red (F) or blue (T). 
The latter costs one more extra tile type, and it will turn out unaffordable later. 

\begin{figure}[tb]
\begin{center}
\begin{tikzpicture}

\bluetile{0}{1.05}{A}{$a$}{0}{$a$}{0};
\bluetile{2}{1.05}{B}{$a$}{0}{$b$}{1};
\bluetile{4}{1.05}{C}{$b$}{1}{$a$}{0};

\bluetile{8.4}{2.1}{}{}{}{}{0}; \node at (8.4, 2.1) {A/C};
\bluetile{8.4}{1.4}{B}{}{}{}{};  
\bluetile{7.7}{0.7}{B}{}{}{}{}; \bluetile{8.4}{0.7}{C}{}{}{}{0}; 
\bluetile{7}{0}{B}{}{}{}{}; \bluetile{7.7}{0}{C}{}{}{}{}; 

\draw[dotted, very thick, ->] (7, -0.7) node[below] {$b$} -- (7, 0) -- (7.7, 0) -- (7.7, 0.7) -- (8.4, 0.7) -- (8.4, 1.4) -- (9.1, 1.4) node[right] {1}; 

\end{tikzpicture}
\end{center}
\caption{
{\rm (Left)} Sole set of 3 cyan tile types with which one can self-assemble the subpattern {\tt LB4}. 
{\rm (Right)} $B$ and $C$ tiles deliver a signal via $b$ and 1 glues in a zigzag manner toward northeast. 
}
\label{fig:spurious_cyan3}
\end{figure}

\begin{lemma}\label{lem:LB4}
	If a directed RTAS with 21 tile types uniquely self-assembles a pattern including {\tt GADGET}, then it has either
	\begin{enumerate}
	\item	at least 4 cyan tile types, or
	\item	the 3 cyan tile types shown in Figure~\ref{fig:spurious_cyan3}, 1 red (F) tile type, 1 blue (T) tile type, and 2 tile types whose color is either red (F) or blue (T). 
	\end{enumerate}
\end{lemma}

Our proof of this lemma is so technical that presenting it at this point may distract the reader's attention from the essence of the reduction. 
Its proof is in Section~\ref{subsec:proof_lemLB4}. 
For the sake of argument to deny the second choice later, we briefly observe how the 3 cyan tiles $A, B, C$ in the choice deliver signals. 
As shown in Figure~\ref{fig:spurious_cyan3} (Right), $B$ and $C$ tiles alternately attach and deliver signals in a zigzag manner. 
Note that they cannot expose two 1 glues to the east consecutively; a 1 glue is vertically sandwiched by 0 glues. 
This is the essential defect not to let {\tt GADGET} assemble as long as the second choice is made. 

Among the 7 uncolored tile types, the first choice in Lemma~\ref{lem:LB4} draws 3 of them by cyan, while the second choice draws 4 of them. 
Now we will see that, not depending on which choice was made, we must draw one of the uncolored tile types by {\tt CE} and another by {\ce} or yellow. 
Just above {\tt LB4}, we find six yellow positions stacked vertically with a {\ce} position on top of them, and to their west is a pillar of {\tt CE}'s. 
Note that we have colored only one tile type by yellow so far, and there are at most 4 tile types left uncolored. 
With only one {\ce} type, no directed RTAS could put a {\ce} tile at the top of the six yellow tiles. 
One uncolored tile type is to be colored by {\ce}. 

The next lemma suggests that at least one of the uncolored tile types must be painted with either {\tt CE} or yellow. 
Its proof is in Section~\ref{subsec:proof_lemlb3}. 

	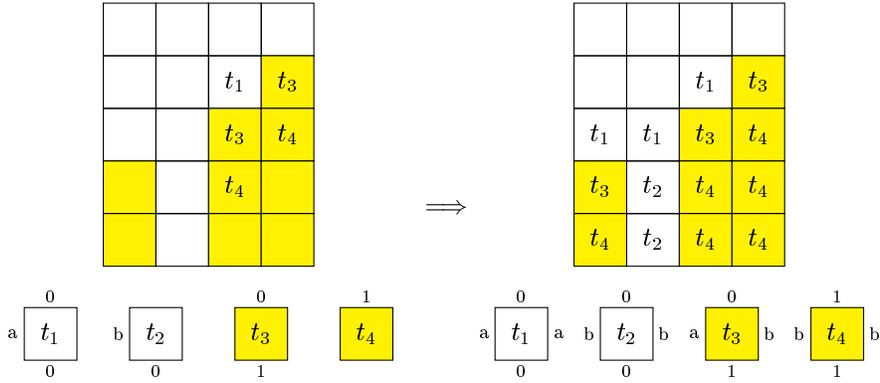
\begin{figure}[tb]
	\begin{minipage}{0.45\linewidth}
	\begin{center}
	\beginpgfgraphicnamed{spurious_CEyellow1}
	\begin{tikzpicture}

		\node[tile] at (0, 2.8) {};	\node[tile] at (0.7, 2.8) {};	\node[tile] at (1.4, 2.8) {};		\node[tile] at (2.1, 2.8) {};
		\node[tile] at (0, 2.1) {};	\node[tile] at (0.7, 2.1) {};	\node[tile] at (1.4, 2.1) {$t_1$};	\node[yellowtile] at (2.1, 2.1) {$t_3$};
		\node[tile] at (0, 1.4) {};	\node[tile] at (0.7, 1.4) {};	\node[yellowtile] at (1.4, 1.4) {$t_3$};\node[yellowtile] at (2.1, 1.4) {$t_4$};
		\node[yellowtile] at (0, 0.7) {};\node[tile] at (0.7, 0.7) {};	\node[yellowtile] at (1.4, 0.7) {$t_4$};\node[yellowtile] at (2.1, 0.7) {};
		\node[yellowtile] at (0, 0) {};	\node[tile] at (0.7, 0) {};	\node[yellowtile] at (1.4, 0) {};	\node[yellowtile] at (2.1, 0) {};

		\whitetile{-1.05}{-1.25}{$t_1$}{0}{a}{0}{};
		\whitetile{0.35}{-1.25}{$t_2$}{}{b}{0}{};
		\yellowtile{1.75}{-1.25}{$t_3$}{0}{}{1}{};
		\yellowtile{3.15}{-1.25}{$t_4$}{1}{}{}{};

	\end{tikzpicture}
	\endpgfgraphicnamed
	\end{center}
	\end{minipage}
	$\Longrightarrow$
	\begin{minipage}{0.45\linewidth}
	\begin{center}
	\beginpgfgraphicnamed{spurious_CEyellow2}
	\begin{tikzpicture}

		\node[tile] at (0, 2.8) {};	\node[tile] at (0.7, 2.8) {};	\node[tile] at (1.4, 2.8) {};		\node[tile] at (2.1, 2.8) {};
		\node[tile] at (0, 2.1) {};	\node[tile] at (0.7, 2.1) {};	\node[tile] at (1.4, 2.1) {$t_1$};	\node[yellowtile] at (2.1, 2.1) {$t_3$};
		\node[tile] at (0, 1.4) {$t_1$};\node[tile] at (0.7, 1.4) {$t_1$};\node[yellowtile] at (1.4, 1.4) {$t_3$};\node[yellowtile] at (2.1, 1.4) {$t_4$};
		\node[yellowtile] at (0, 0.7) {$t_3$};\node[tile] at (0.7, 0.7) {$t_2$};\node[yellowtile] at (1.4, 0.7) {$t_4$};\node[yellowtile] at (2.1, 0.7) {$t_4$};
		\node[yellowtile] at (0, 0) {$t_4$};\node[tile] at (0.7, 0) {$t_2$};\node[yellowtile] at (1.4, 0) {$t_4$};\node[yellowtile] at (2.1, 0) {$t_4$};

		\whitetile{-1.05}{-1.25}{$t_1$}{0}{a}{0}{a};
		\whitetile{0.35}{-1.25}{$t_2$}{0}{b}{0}{b};
		\yellowtile{1.75}{-1.25}{$t_3$}{0}{a}{1}{b};
		\yellowtile{3.15}{-1.25}{$t_4$}{1}{b}{1}{b};

	\end{tikzpicture}
	\endpgfgraphicnamed
	\end{center}
	\end{minipage}
	\caption{
		This subpattern of {\tt GADGET} can assemble in this way using the 2 {\tt CE} tile types and 2 yellow tile types shown here. 
	}
	\label{fig:lb3-bc}
	\end{figure}

\begin{lemma}\label{lem:lb3}
	If a directed RTAS with at most 21 tile types uniquely self-assembles a pattern including {\tt GADGET}, then it contains at least 2 {\ce} tile types and the sum of the number of {\tt CE} tile types and the number of yellow tile types is at least 4. 
	Moreover, if it contains exactly 2 {\tt CE} tile types $t_1, t_2$ and exactly 2 yellow tile types $t_3, t_4$, then these four tile types are labelled as depicted at the right bottom of Figure~\ref{fig:lb3-bc}. 
\end{lemma}

This lemma suggests one non-isomorphic way to paint/label 4 tile types so that resulting tiles uniquely self-assemble the pattern in Figure~\ref{fig:lb3-bc}, which is a subpattern of {\tt GADGET}, found in Figure~\ref{fig:gadget1}. 
This way, however, shall be proven improper in order for a directed RTAS with 21 tile types to self-assemble the whole {\tt GADGET} in the end. 
In any case, this lemma implies that among the at most\footnote{If the second option in Lemma~\ref{lem:LB4} is chosen, then there are only 2 uncolored tile types at this point.} 3 uncolored tile types, one must be painted either {\tt CE} (expected) or yellow (unexpected). 

Let us summarize visually how the 21 tile types have been painted so far, where a dotted square indicates an uncolored tile type:   

\begin{center}
\scalebox{0.95}{
\begin{tikzpicture}

\node[bluetile] at (0, 2) {}; \node[bluetile] at (1, 2) {}; \node[bluetile] at (2, 2) {}; \node[tile] at (3, 2) {\tt CE}; \node[tile] at (4, 2) {\tt CE}; 
\node[graytile] at (0, 1) {\scriptsize {\tt Init}}; \node[tile, fill=black] at (1, 1) {}; \node[tile] at (2, 1) {}; \node[tile, fill=red, text=white] at (3, 1) {\tt F}; \node[tile, fill=yellow] at (4, 1) {};
\node[graytile] at (0, 0) {\scriptsize {\tt Init}}; \node[tile, fill=black] at (1, 0) {};\node[tile] at (2, 0) {}; \node[tile, fill=blue, text=white] at (3, 0) {\tt T}; \node[tile, fill=green] at (4, 0) {\footnotesize {\tt Sat}}; 
\draw[fill=black] (1, 2) -- (0.65, 1.65) -- (1.35, 1.65) -- (1.35, 2.35) -- cycle; 
\draw[fill=white] (2, 2) -- (1.65, 1.65) -- (2.35, 1.65) -- (2.35, 2.35) -- cycle; 

\node at (5, 1) {\large +}; 

\node[tile] at (6, 1) {{\tt CE}};
\draw[fill=yellow] (6, 1) -- (5.65, 0.65) -- (6.35, 0.65) -- (6.35, 1.35) -- cycle;

\node at (7, 1) {\large +};

\node[bluetile] at (8, 2) {}; \node[bluetile] at (9, 2) {}; \node[bluetile] at (10, 2) {}; \node[tile, dashed] at (11, 2) {}; \node[tile, dashed] at (12, 2) {};
\node at (10, 1) {\large OR};
\node[bluetile] at (8, 0) {}; \node[bluetile] at (9, 0) {}; \node[tile, fill=red, text=white] at (10, 0) {}; \draw[fill=blue] (10, 0) -- (9.65, -0.35) -- (10.35, -0.35) -- (10.35, 0.35) -- cycle; \node[tile, fill=red, text=white] at (11, 0) {}; \draw[fill=blue] (11, 0) -- (10.65, -0.35) -- (11.35, -0.35) -- (11.35, 0.35) -- cycle; \node[tile, dashed] at (12, 0) {};

\draw(7.5, 2.5) -- node[above] {First option of Lemma~\ref{lem:LB4}} (12.5, 2.5) -- (12.5, 1.5) -- (7.5, 1.5) -- cycle; 
\draw(7.5, -0.5) -- node[below] {Second option of Lemma~\ref{lem:LB4}} (12.5, -0.5) -- (12.5, 0.5) -- (7.5, 0.5) -- cycle;

\end{tikzpicture} 
}
\end{center}

We now exclude the second choice of Lemma~\ref{lem:LB4}. 
For the sake of contradiction, suppose that with this spurious option, a directed RTAS could self-assemble {\tt GADGET}. 
Then as of now, only one tile type remains uncolored, and hence, one of the following statements must hold: 
\begin{itemize}
\item	there are only 1 {\tt DGNL}-white and 2 white tile types; 
\item	there are only 1 {\tt DGNL}-black and 2 black tile types. 
\end{itemize} 

The 16 subpatterns of {\tt GADGET} in Figure~\ref{fig:lb2-DGNL} play a role in denying the first statement. 
Consider the task for the RTAS to assemble these 16 subpatterns with only 1 {\tt DGNL}-white and 2 white tile types. 
Their assemblies are trivially identical at the main diagonal consisting of four {\tt DGNL}-white positions (1, 1) - (4, 4). 
Recall that the 2 white tile types have distinct west glues (Lemma~\ref{lem:exactly2}). 
Hence, all white positions on the first diagonal below the main diagonal are filled with tiles of the same type. 
This argument works also for the second and third diagonal below the main one. 
As a result, the 16 assemblies are identical with respect to their fourth column from the left. 
The RTAS being directed, this means that types of tiles at the bottom of the rightmost two columns ({\tt Init} and red(F)/blue(T)) completely determine which of the 16 subpatterns emerges. 
However, even with painting the last uncolored tile type with {\tt Init}, at most $12 (= 3 \times 4)$ combinations of types would be possible, that is, four of the 16 subpatterns could never assemble, a contradiction. 
Likewise, the second statement is denied by the black analogue of these 16 subpatterns. 
The second option of Lemma~\ref{lem:LB4} has been thus excluded. 
As a result, the 21 tile types have been colored partially as follows: 

\begin{center}
\scalebox{0.95}{
\begin{tikzpicture}

\node[graytile] at (0, 1) {\scriptsize {\tt Init}}; \node[tile, fill=black] at (1, 1) {}; \node[tile] at (2, 1) {}; \node[tile, fill=red, text=white] at (3, 1) {\tt F}; \node[tile, fill=yellow] at (4, 1) {}; \node[bluetile] at (5, 1) {}; \node[bluetile] at (6, 1) {}; \node[bluetile] at (7, 1) {}; \node[bluetile] at (8, 1) {}; 
\node[graytile] at (0, 0) {\scriptsize {\tt Init}}; \node[tile, fill=black] at (1, 0) {};\node[tile] at (2, 0) {}; \node[tile, fill=blue, text=white] at (3, 0) {\tt T}; \node[tile, fill=green] at (4, 0) {\footnotesize {\tt Sat}}; \node[bluetile] at (5, 0) {}; \blacklowertriangle{5}{0}; \node[bluetile] at (6, 0) {}; \whitelowertriangle{6}{0}; \node[tile] at (7, 0) {\tt CE}; \node[tile] at (8, 0) {\tt CE}; 

\node[tile] at (9, 0) {{\tt CE}};
\draw[fill=yellow] (9, 0) -- (8.65, -0.35) -- (9.35, -0.35) -- (9.35, 0.35) -- cycle;

\node at (10, 0.5) {\large +}; 

\node[tile, dashed] at (11, 1) {}; 
\node[tile, dashed] at (11, 0) {};

\end{tikzpicture} 
}
\end{center}

We conclude the coloring by proving that one of the remaining 2 uncolored tile types must be painted {\tt DGNL}-white and the other {\tt DGNL}-black. 

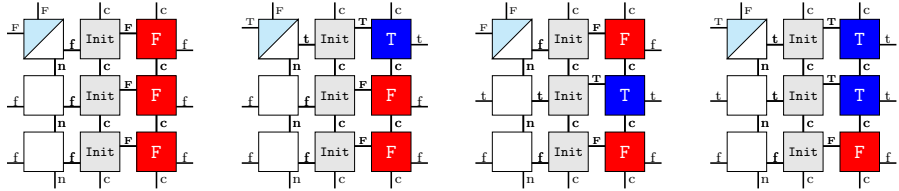
\begin{figure}[htb]
\begin{minipage}{0.225\linewidth}
\scalebox{0.75}{
\begin{tikzpicture}
\DGNLwF{0}{2}; \Initf{1}{2}; \redF{2}{2};
\whitef{0}{1}; \Initf{1}{1}; \redF{2}{1};
\whitef{0}{0}; \Initf{1}{0}; \redF{2}{0}; 
\end{tikzpicture}
}
\end{minipage}
\begin{minipage}{0.0125\linewidth}
\ \\
\end{minipage}
\begin{minipage}{0.225\linewidth}
\scalebox{0.75}{
\begin{tikzpicture}
\DGNLwT{0}{2}; \Initt{1}{2}; \blueT{2}{2};
\whitef{0}{1}; \Initf{1}{1}; \redF{2}{1};
\whitef{0}{0}; \Initf{1}{0}; \redF{2}{0}; 
\end{tikzpicture}
}
\end{minipage}
\begin{minipage}{0.0125\linewidth}
\ \\
\end{minipage}
\begin{minipage}{0.225\linewidth}
\scalebox{0.75}{
\begin{tikzpicture}
\DGNLwF{0}{2}; \Initf{1}{2}; \redF{2}{2};
\whitet{0}{1}; \Initt{1}{1}; \blueT{2}{1};
\whitef{0}{0}; \Initf{1}{0}; \redF{2}{0}; 
\end{tikzpicture}
}
\end{minipage}
\begin{minipage}{0.0125\linewidth}
\ \\
\end{minipage}
\begin{minipage}{0.225\linewidth}
\scalebox{0.75}{
\begin{tikzpicture}
\DGNLwT{0}{2}; \Initt{1}{2}; \blueT{2}{2};
\whitet{0}{1}; \Initt{1}{1}; \blueT{2}{1};
\whitef{0}{0}; \Initf{1}{0}; \redF{2}{0}; 
\end{tikzpicture}
}
\end{minipage}
\caption{Parts of 4 instances of the template in Figure~\ref{fig:lb2-DGNL} (Left).}
\label{fig:subsubpatterns}
\end{figure}

For the sake of contradiction, suppose only one {\tt DGNL}-white tile type available. 
Among the 16 instances of the template shown in Figure~\ref{fig:lb2-DGNL} (Left), consider the eight of them whose right bottom corner is {\tt Init}-red(F). 
With only two white tile types, as argued just above, the type of the {\tt Init} and red(F) tiles attaching there completely determines which of the possible 8 red(F)-blue(T) patterns assembles above. 
However, no matter how we paint the remaining 2 uncolored tile types, the number of combinations of {\tt Init} tile types and red(F) tile types cannot exceed 6, and hence, at least 2 of the 8 subpatterns could not be assembled, a contradiction. 
Hence, we cannot do without coloring one more tile type by white. 
Only one tile type being uncolored now, either there is only one red tile type or there is only one blue tile type. 
Consider the first case. 
See Figure~\ref{fig:subsubpatterns} for parts of four instances. 
At their northeast corner, we find all of FF, FT, TF, and TT (they are vertically aligned), and which of them appears is completely determined by how the downward-diagonal consisting of the top-left {\tt DGNL}-white position, middle {\tt Init} position, and bottom-right red (F) position assembles. 
For that, 4 {\tt Init} tile types are required, but there are at most 3 {\tt Init} tile types available, a contradiction. 
The argument based on the blue analogues of the subpatterns leads us to the same contradiction, provided there is only one blue tile type. 

Consequently, one of the 2 uncolored tile types is to be painted {\tt DGNL}-white. 
Based on the 16 instances of the black analogue of the template, on which white and {\tt DGNL}-white positions are painted rather black and {\tt DGNL}-black, respectively, the argument above creates the need for one more {\tt DGNL}-black tile type. 

We have proved that if a directed RTAS with at most 21 tile types uniquely self-assembles a pattern including {\tt GADGET}, then the tile types must be colored as:  
\vspace*{2mm} 

\begin{center}
\scalebox{0.95}{
\begin{tikzpicture}

\node[graytile] at (0, 1) {\scriptsize {\tt Init}}; \node[tile, fill=black] at (1, 1) {}; \node[tile] at (2, 1) {}; \node[tile, fill=red, text=white] at (3, 1) {\tt F}; \node[tile, fill=yellow] at (4, 1) {}; \node[bluetile] at (5, 1) {}; \node[bluetile] at (6, 1) {}; \node[bluetile] at (7, 1) {}; \node[bluetile] at (8, 1) {}; \node[tile] at (9, 1) {\tt CE}; 
\node[graytile] at (0, 0) {\scriptsize {\tt Init}}; \node[tile, fill=black] at (1, 0) {};\node[tile] at (2, 0) {}; \node[tile, fill=blue, text=white] at (3, 0) {\tt T}; \node[tile, fill=green] at (4, 0) {\footnotesize {\tt Sat}}; \node[bluetile] at (5, 0) {}; \blacklowertriangle{5}{0}; \node[bluetile] at (6, 0) {}; \blacklowertriangle{6}{0}; \node[bluetile] at (7, 0) {}; \whitelowertriangle{7}{0}; \node[bluetile] at (8, 0) {}; \whitelowertriangle{8}{0}; \node[tile] at (9, 0) {\tt CE}; 

\node[tile] at (10, 0) {{\tt CE}};
\draw[fill=yellow] (10, 0) -- (9.65, -0.35) -- (10.35, -0.35) -- (10.35, 0.35) -- cycle;

\end{tikzpicture} 
}
\end{center}
\vspace*{2mm}

\noindent
We will see the color of the last one be determined {\tt CE} in the next subsection. 

	\subsubsection{Glue assignment}

Having colored the 21 tile types almost completely, now we will proceed to the issue of glue assignment; how should we assign glues to the 21 tile types so that the directed RTAS with the resulting tile type set can uniquely self-assemble a pattern including {\tt GADGET}? 

It is easy to determine the glue assignment of tile types which do not share their color with another tile type, that is, the {\tt Sat}, red (F), and blue (T) tile types. 
Let us denote these tile types by $t_{\tt Sat}$, $t_{\tt F}$, and $t_{\tt T}$, respectively. 
All {\tt Sat}, red, and blue positions on {\tt GADGET} are filled with $t_{\tt Sat}$, $t_{\tt F}$, and $t_{\tt T}$ tiles, respectively. 
On {\tt GADGET}, red and blue positions are found vertically stacked so that $t_{\tt F}(\north) = t_{\tt F}(\south) = t_{\tt T}(\north) = t_{\tt T}(\south) = {\tt c}$ for some glue ${\tt c}$. 
For the sake of directedness, this lets $t_{\tt F}(\west) = {\tt F}$ and $t_{\tt T}(\west) = {\tt T}$ for some distinct glues ${\tt F}, {\tt T}$. 
A {\tt Sat} position is found to the north and to the west of a red position (see Figure~\ref{fig:gadget2}) so that $t_{\tt Sat}(\south) = t_{\tt F}(\north) = {\tt c}$ and $t_{\tt Sat}(\east) = t_{\tt F}(\west) = {\tt F}$. 
Sharing the south glue with $t_{\tt F}$ and $t_{\tt T}$, $t_{\tt Sat}(\west)$ must be different from $t_{\tt F}(\west)$ or $t_{\tt T}(\west)$ for the sake of directedness; let $t_{\tt Sat}(\west) = {\tt s}$ for some new glue ${\tt s}$. 
Their glues have been determined (partially) as follows: 
\vspace*{2mm}
\begin{center}
\begin{tikzpicture}
	\node[tile, fill=red, text=white] at (0, 0) {{\tt F}}; \glnc{0}{0}; \glwF{0}{0}; \glsc{0}{0};
	\node[tile, fill=blue, text=white] at (2, 0) {{\tt T}}; \glnc{2}{0}; \glwT{2}{0}; \glsc{2}{0};
	\node[tile, fill=green] at (4, 0) {\footnotesize {\tt Sat}}; \glws{4}{0}; \glsc{4}{0}; \gleF{4}{0};
\end{tikzpicture}
\end{center}

Next we see how the 2 {\tt Init} tile types, which we denote by $t_{\tt InitF}$ and $t_{\tt InitT}$, are assigned with glues. 
See Figure~\ref{fig:gadget2}, where we find a column of {\tt Init} positions sandwiched by two columns of red and blue positions, at which $t_{\tt F}$ and $t_{\tt T}$ tiles attach, respectively. 
Thus, without loss of generality (w.l.o.g.), the type of tile at an {\tt Init} position between red positions is $t_{\tt InitF}$ while the type of tile at an {\tt Init} position between blue positions is $t_{\tt InitT}$. 
This implies $t_{\tt InitF}(\north) = t_{\tt InitF}(\south) = t_{\tt InitT}(\north) = t_{\tt InitT}(\south)$. 
This glue is actually {\tt c} because a red position is found on top of the {\tt Init} column. 
For the sake of directedness, we need to introduce new glues ${\tt f}, {\tt t} \neq {\tt s}, {\tt F}, {\tt T}$ as respective west glues of $t_{\tt InitF}$ and $t_{\tt InitT}$. 
Now the three horizontally-adjacent positions red-{\tt Init}-red imply $t_{\tt F}(\east) = t_{\tt InitF}(\west) = {\tt f}$ and $t_{\tt InitF}(\east) = t_{\tt F}(\west) = {\tt F}$. 
Similarly, we get $t_{\tt T}(\east) = t_{\tt InitT}(\west) = {\tt t}$ and $t_{\tt InitT}(\east) = t_{\tt T}(\west) = {\tt T}$. 
The glues of the 5 tile types have been thus determined (partially) as follows: 
\vspace*{2mm}
\begin{center}
\begin{tikzpicture}
	\redF{0}{0};
	\blueT{2}{0};
	\node[tile, fill=green] at (4, 0) {\footnotesize {\tt Sat}}; \glws{4}{0}; \glsc{4}{0}; \gleF{4}{0};
	\Initf{6}{0}; 
	\Initt{8}{0};
\end{tikzpicture}
\end{center}

The same argument is applied to a white column next to an {\tt Init} column which is sandwiched by red(F)/blue(T) columns (see Figure~\ref{fig:gadget2}) to assign the two white tile types $t_{\tt wf}, t_{\tt wt}$ with glues as $t_{\tt wf}(\west) = t_{\tt wf}(\east) = {\tt f}$, $t_{\tt wt}(\west) = t_{\tt wt}(\east) = {\tt t}$, and $t_{\tt wf}(\south) = t_{\tt wt}(\south) = {\tt n}$ for some glue ${\tt n}$, which must differ from {\tt c} for directedness. 
With these, the black column in Figure~\ref{fig:gadget2} enforces the following glue assignment to the two black tile types $t_{\tt bf}$ and $t_{\tt bt}$ according to the same argument: 
\vspace*{2mm}
\begin{center}
\begin{tikzpicture}
	\whitef{0}{0}; \whitet{2}{0}; \blackf{4}{0}; \blackt{6}{0};
\end{tikzpicture}
\end{center}
where the glue {\tt v} must differ from {\tt n}. 

As seen above, finding a color with which exactly one tile type is painted is useful for the glue assignment. 
Recall the tile type whose color was not determined but just narrowed down to be either {\tt CE} or yellow. 
We now prove that it must be painted {\tt CE}, and the tile type set turns out to contain only 1 yellow tile type. 
See the bottom left corner of Figure~\ref{fig:gadget1}, where there is the pattern red(F)-\ce-{\tt Sat}. 
With only two {\ce} tile types, this pattern would imply the contradictory equation ${\tt f} = {\tt s}$. 
This is because of Lemma~\ref{lem:lb3}, which suggests that with only two {\tt CE} tile types available, {\tt CE} tiles would just let a 1-bit signal pass through from west to east. 
Thus, we must draw the tile type by {\tt CE}. 

Now the tile type set contains only one yellow tile type $t_{\tt y}$. 
See Figure~\ref{fig:gadget1}, in which yellow positions are adjacent to each other horizontally and vertically and a yellow position is to the west of a {\tt Sat} position at the bottom of {\tt LB4}. 
Thus, we have $t_{\tt y}(\west) = t_{\tt y}(\east) = t_{\tt Sat}(\west) = {\tt s}$ and $t_{\tt y}(\north) = t_{\tt y}(\south)$, and moreover, $t_{\tt y}(\south) \neq t_{\tt Sat}(\south) = {\tt c}$ must hold since they have the same west glue {\tt s}; let $t_{\tt y}(\south) = {\tt T}$ for some glue ${\tt T} \neq {\tt c}$. 
\vspace*{2mm}
\begin{center}
\begin{tikzpicture}
	\yellows{0}{0};
\end{tikzpicture}
\end{center}
It must be noted that at this point, we cannot exclude the possibility that the south glue {\tt T} is equal to {\tt n} or {\tt v}. 
It is not until the glue assignment of all the 21 tile types is completely determined at the end of this section that {\tt T} is distinguished from them.  

Let us now shift our attention to the glue assignment to the 3 {\tt CE} tile types. 
Since there is a {\tt CE} position horizontally sandwiched by two yellow positions in Figure~\ref{fig:gadget1}, one {\tt CE} tile type, say $t_{\tt CEss}$, has {\tt s} glues on its west and east edges. 
Thus, its south glue must be different from the south glue of yellow tile type or from that of {\tt Sat} tile type; let $t_{\tt CEss}(\south) = {\tt F}$ for some glue ${\tt F} \neq {\tt T}, {\tt c}$ (note that as the glue {\tt T}, the distinction of {\tt F} from {\tt n} or {\tt v} will not be made until the end of this section). 
In Figure~\ref{fig:gadget2}, we find positions red(F)-{\tt CE}-{\tt Init}-red. 
At the red positions, $t_{\tt F}$ tiles attach, so the type of tile attaching at the {\tt Init} position is $t_{\tt InitF}$. 
Thus, the {\tt CE} tile there must have the glue {\tt f} at both the west and east sides, and hence, cannot be of type $t_{\tt CEss}$. 
Let us denote its type by $t_{\tt CEff}$; then $t_{\tt CEff}(\west) = t_{\tt CEff}(\east) = {\tt f}$. 
\vspace*{2mm}
\begin{center}
\begin{tikzpicture}
	\node[tile] at (0, 0) {{\tt CE}}; \glws{0}{0}; \glsF{0}{0}; \gles{0}{0};
	\node[tile] at (2, 0) {{\tt CE}}; \glwf{2}{0}; \glef{2}{0};
\end{tikzpicture}
\end{center}

On the top row of {\tt LB4} back in Figure~\ref{fig:gadget1}, there is a pattern red(F)-{\tt CE}-{\tt CE}-yellow. 
The types of tiles at the red and yellow positions are $t_{\tt F}$ and $t_{\tt yellow}$, respectively, and $t_{\tt F}(\east) = {\tt f}$ while $t_{\tt yellow}(\west) = {\tt s}$. 
In order to assemble this pattern, therefore, we need the third {\tt CE} tile type $t_{\tt CEfs}$ with $t_{\tt CEfs}(\west) = {\tt f}$ and $t_{\tt CEfs}(\east) = {\tt s}$. 
\vspace*{2mm}
\begin{center}
\begin{tikzpicture}
	\node[tile] at (0, 0) {{\tt CE}}; \glws{0}{0}; \glsF{0}{0}; \gles{0}{0};
	\node[tile] at (2, 0) {{\tt CE}}; \glwf{2}{0}; \glef{2}{0};
	\node[tile] at (4, 0) {{\tt CE}}; \glwf{4}{0}; \gles{4}{0};
\end{tikzpicture}
\end{center}

Their north and south glues must be determined now. 
See Figure~\ref{fig:gadget1}, where we find yellow-{\tt CE}-yellow positions self-stacked vertically. 
{\ce} tiles attaching there must have {\tt s} glues on their west and east edges, and hence, are of type $t_{\tt CEss}$. 
Thus, $t_{\tt CEss}(\north) = t_{\tt CEss}(\south) = {\tt F}$. 
In the same figure, we find a {\tt CE} position whose east and south neighbors are yellow. 
A $t_{\tt CEss}$ tile cannot attach there due to its south glue mismatch; neither can a $t_{\tt CEff}$ tile due to its east glue mismatch. 
The remaining type $t_{\tt CEfs}$ must be assigned with glues properly so as for a $t_{\tt CEfs}$ tile to attach there. 
Thus, $t_{\tt CEfs}(\south) = {\tt T}$. 
\vspace*{2mm}
\begin{center}
\begin{tikzpicture}
	\CEss{0}{0};
	\node[tile] at (2, 0) {{\tt CE}}; \glwf{2}{0}; \glef{2}{0};
	\node[tile] at (4, 0) {{\tt CE}}; \glwf{4}{0}; \glsT{4}{0}; \gles{4}{0};
\end{tikzpicture}
\end{center}
\vspace*{2mm}
In Figure~\ref{fig:gadget1}, we can see three vertically-stacked {\tt CE} positions sandwiched horizontally by yellow positions. 
Hence, $t_{\tt CEss}$ tiles attach there. 
Focus on the {\tt CE} position above them, and let us denote it by $(x, y)$. 
Since its eastern neighbor is yellow, the type of {\tt CE} tile attaching there must be also $t_{\tt CEss}$. 
We see a $t_{\tt CEff}$ tile attach to its north neighbor position $(x, y+1)$ and a $t_{\tt CEfs}$ tile attach to its western neighbor position $(x-1, y)$. 
For these last two placements, it suffices to observe that at all {\tt CE} positions just above the stair-like yellow positions is $t_{\tt CEfs}$ and tiles at all the consecutive {\tt CE} positions to the west of $t_{\tt CEfs}$ must be of type $t_{\tt CEff}$. 
Now, around $(x, y)$, tiles assemble as: 
\[
\begin{array}{ccccc}
	t_{\tt CEff} & t_{\tt CEff} & t_{\tt CEff} & t_{\tt CEfs} & \mbox{\textcolor{yellow}{$\blacksquare$}}\\
	t_{\tt CEff} & t_{\tt CEfs} & \fbox{$t_{\tt CEss}$} & \mbox{\textcolor{yellow}{$\blacksquare$}} & \mbox{\textcolor{yellow}{$\blacksquare$}}\\
	t_{\tt CEfs} & \mbox{\textcolor{yellow}{$\blacksquare$}} & t_{\tt CEss} & \mbox{\textcolor{yellow}{$\blacksquare$}} & \mbox{\textcolor{yellow}{$\blacksquare$}} \\
	\mbox{\textcolor{yellow}{$\blacksquare$}} & \mbox{\textcolor{yellow}{$\blacksquare$}} & t_{\tt CEss} & \mbox{\textcolor{yellow}{$\blacksquare$}} & \mbox{\textcolor{yellow}{$\blacksquare$}} \\
	\mbox{\textcolor{yellow}{$\blacksquare$}} & \mbox{\textcolor{yellow}{$\blacksquare$}} & t_{\tt CEss} & \mbox{\textcolor{yellow}{$\blacksquare$}} & \mbox{\textcolor{yellow}{$\blacksquare$}}
\end{array}
\]
where the position $(x, y)$ is indicated by the box. 
Thus, $t_{\tt CEff}(\south) = {\tt F}$, and this in turn gives $t_{\tt CEff}(\north) = t_{\tt CEfs}(\north) = {\tt F}$. 
Now the glues of the 3 {\tt CE} tile types have been determined completely as follows: 
\vspace*{2mm}
\begin{center}
\begin{tikzpicture}
	\CEss{0}{0};
	\CEff{2}{0};
	\CEfs{4}{0};
\end{tikzpicture}
\end{center}
They have one useful property. 
\begin{property}\label{propty:yellow_CEss}
	Let $x, y \in \mathbb{N}_0$ and $d \ge 1$. 
	If ${\tt GADGET}(x, y)$ is yellow, while for all $1 \le i \le d$, ${\tt GADGET}(x{+}i, y)$ is {\tt CE}, then the type of tile at $(x{+}d, y)$ is $t_{\tt CEss}$. 
\end{property}

Before proceeding to the glue assignment of the remaining 8 tile types (4 cyan, 2 {\tt DGNL}-white, and 2 {\tt DGNL}-black), we determine the north glue of $t_{\tt Sat}$. 
In Figure~\ref{fig:gadget2}, there is one {\tt Sat} position whose north neighbor is {\tt CE} and whose northwestern neighbor is yellow. 
Due to Property~\ref{propty:yellow_CEss}, the type of tile attaching at this {\tt CE} position is $t_{\tt CEss}$, and hence, $t_{\tt Sat}(\north) = t_{\tt CEss}(\south) = {\tt F}$. 
Let us present all the 13 tile types whose glues are completely determined so far: 
\vspace*{2mm}
\begin{center}
\begin{tikzpicture}
	\redF{0}{1.5}; \blueT{1.5}{1.5}; \greenSat{3}{1.5}; \yellows{4.5}{1.5}; \Initf{6}{1.5}; \Initt{7.5}{1.5};
	\whitef{0}{0}; \whitet{1.5}{0}; \blackf{3}{0}; \blackt{4.5}{0}; \CEss{6}{0}; \CEff{7.5}{0}; \CEfs{9}{0};
\end{tikzpicture}
\end{center}

Let us determine the glues of 4 cyan tile types. 
First, see the tenth column from the right in Figure~\ref{fig:gadget2}, on which there is a cyan position surrounded by {\tt CE}'s, {\tt Sat}, and red positions. 
Due to Property~\ref{propty:yellow_CEss} and the fact that the north glue of any {\tt CE} tile is {\tt F}, the tile attaching there must be assigned with {\tt F} glues on all of its four sides as: 
\begin{center}
\begin{tikzpicture}
	\blueff{0}{0}; \node at (0, 0) {$t_{\tt sbFF}$}; 
\end{tikzpicture}
\end{center}
Let us denote this type by $t_{\tt sbFF}$. 

See the horizontal tandem of cyan positions just above {\tt LB4} in Figure~\ref{fig:gadget1}. 
At its first and second positions, $t_{\tt sbFF}$ tiles attach. 
The tile attaching at the third one must have {\tt F} glue on its west side and {\tt T} glues on its north and south sides, and hence, it is not of type $t_{\tt sbFF}$. 
Let us denote its type by $t_{\tt sbFT}$. 
\begin{center}
\begin{tikzpicture}
	\blueff{0}{0}; \node at (0, 0) {$t_{\tt sbFF}$}; 
	\bluetile{2}{0}{}{}{}{}{}; \glnT{2}{0}; \glwF{2}{0}; \glsT{2}{0}; \node at (2, 0) {$t_{\tt sbFT}$}; 
\end{tikzpicture}
\end{center}

At the southwestern corner of {\tt LB4}, a $t_{\tt sbFF}$ tile attaches, and hence, the tile attaching to its north must have west glue {\tt T} and south glue {\tt F}. 
Thus, the tile is of type neither $t_{\tt sbFF}$ nor $t_{\tt sbFT}$; let it be $t_{\tt sbTF}$. 
Let us denote the fourth cyan tile type by $t_{\tt sbTT}$. 
\begin{center}
\begin{tikzpicture}
	\blueff{0}{0}; \node at (0, 0) {$t_{\tt sbFF}$}; 
	\bluetile{2}{0}{}{}{}{}{}; \glnT{2}{0}; \glwF{2}{0}; \glsT{2}{0}; \node at (2, 0) {$t_{\tt sbFT}$}; 
	\bluetile{4}{0}{}{}{}{}{}; \glwT{4}{0}; \glsF{4}{0}; \node at (4, 0) {$t_{\tt sbTF}$}; 
	\bluetile{6}{0}{}{}{}{}{}; \node at (6, 0) {$t_{\tt sbTT}$}; 
\end{tikzpicture}
\end{center}

We claim that $t_{\tt sbTT}(\west) = {\tt T}$. 
Suppose not; then $t_{\tt sbTF}$ would be the sole cyan tile type whose west glue is {\tt T}. 
See Figure~\ref{fig:gadget2}, where there is a pattern blue--white--{\tt Init}--cyan--blue, and it assembles as $t_{\tt T}$--$t_{\tt wt}$--$t_{\tt InitT}$--cyan--$t_{\tt T}$. 
This arises a need for a cyan tile type both of whose west and east glues are {\tt T}. 
Hence, $t_{\tt sbTF}(\east) = {\tt T}$. 
Then at the northwestern corner of {\tt LB4}, two cyan tiles of this type attach and expose glues of same kind to the north. 
Accordingly, the assembly at the {\tt CE} positions to their north is either $t_{\tt CEff}^2$ or $t_{\tt CEss}^2$, but in any case, it would cause a glue mismatch either with the red tile to the west or yellow tile to the east, a contradiction. 
The claim $t_{\tt sbTT}(\west) = {\tt T}$ has been verified. 

On the ninth column of Figure~\ref{fig:gadget2}, we find a cyan position surrounded by yellow, {\tt Sat}, and red positions from the north, west, and east, respectively. 
The north, west, and east glues of tile attaching there are required to be {\tt T}, {\tt F}, {\tt F}, respectively. 
Because $t_{\tt sbTT}(\west) = {\tt T}$, its type is $t_{\tt sbFT}$ so that we get $t_{\tt sbFT}(\east) = {\tt F}$. 

\begin{center}
\begin{tikzpicture}
	\blueff{0}{0}; \node at (0, 0) {$t_{\tt sbFF}$}; 
	\blueft{2}{0}; \node at (2, 0) {$t_{\tt sbFT}$}; 
	\bluetile{4}{0}{}{}{}{}{}; \glwT{4}{0}; \glsF{4}{0}; \node at (4, 0) {$t_{\tt sbTF}$}; 
	\bluetile{6}{0}{}{}{}{}{}; \glwT{6}{0}; \node at (6, 0) {$t_{\tt sbTT}$}; 
\end{tikzpicture}
\end{center}

See the cyan positions below the $t_{\tt sbFT}$ tile. 
The tile attaching just below it must have north glue {\tt T} and west and east glues {\tt F}, and hence, it is also of the type $t_{\tt sbFT}$. 
In this way, we figure out that at the top five positions of this cyan ninth column, $t_{\tt sbFT}$ tiles attach. 
Consider the sixth position from the top. 
The tile attaching there must have {\tt T} glues on its north, west, and east sides. 
Hence, it is of type either $t_{\tt sbTF}$ or $t_{\tt sbTT}$ and has north glue {\tt T}. 
Let us identify another cyan position at which a tile of one of these types must attach and moreover its north glue is required rather to be {\tt F}. 
Due to the requirement of different north glues, cyan tiles attaching at these positions must be of different type. 
Such a cyan position is found at the northeastern corner of {\tt LB4}. 
Its north neighbor is {\tt CE} and east neighbor is blue. 
Due to Property~\ref{propty:yellow_CEss}, the tile attaching there must have north glue {\tt F} and east glue {\tt T}, and hence, is of type either $t_{\tt sbTF}$ or $t_{\tt sbTT}$. 
As such, $t_{\tt sbTF}$ and $t_{\tt sbTT}$ tiles attach at these positions exclusively, and hence, we get: 
\begin{itemize}
\item	$t_{\tt sbTF}(\east) = t_{\tt sbTT}(\east) = {\tt T}$. 
\item	$\{t_{\tt sbTF}(\north), t_{\tt sbTT}(\north)\} = \{{\tt F}, {\tt T}\}$. 
\end{itemize}
The latter means that the north glue of any cyan tile is either {\tt F} or {\tt T}. 
As a result, $t_{\tt sbTT}(\south)$ must be either {\tt F} or {\tt T} because below these positions are cyan positions. 
It actually must be {\tt T} for the sake of directedness. 

\begin{center}
\begin{tikzpicture}
	\blueff{0}{0}; \node at (0, 0) {$t_{\tt sbFF}$}; 
	\blueft{2}{0}; \node at (2, 0) {$t_{\tt sbFT}$}; 
	\bluetile{4}{0}{}{}{}{}{}; \glwT{4}{0}; \glsF{4}{0}; \gleT{4}{0}; \node at (4, 0) {$t_{\tt sbTF}$}; 
	\bluetile{6}{0}{}{}{}{}{}; \glwT{6}{0}; \glsT{6}{0}; \gleT{6}{0}; \node at (6, 0) {$t_{\tt sbTT}$}; 
\end{tikzpicture}
\end{center}

Suppose that the north glue of $t_{\tt sbTF}$ is {\tt T} and that of $t_{\tt sbTT}$ is {\tt F}. 
Using tiles of these types, however, we cannot assemble {\tt LB4}. 
These spurious tile types have the following properties:
\begin{itemize}
\item	{\tt F}/{\tt T} signals are faithfully propagated horizontally; 
\item	An {\tt F}/{\tt T} signal is faithfully propagated from the south to the north when crossing a horizontal {\tt F} signal; 
\item	An {\tt F}/{\tt T} signal is flipped when crossing a horizontal {\tt T} signal. 
\end{itemize}
Due to the first property, the cyan portion of the fourth column of {\tt LB4} exposes the following sequence of glues to the east: ${\tt F} {\tt T}^{2r-1} {\tt F}^{2r-1} {\tt TT}$ (from the bottom to the top), which is the same as the one exposed to the east by the {\tt Init} portion of the second column of {\tt LB4}. 
The north glue {\tt T} of the yellow tile attaching at the bottom of the fifth column of {\tt LB4} crosses an odd number of {\tt T} signals while propagating northward, and turns out to be flipped and exposed as {\tt F} to the top yellow position of the column. 
The {\tt T} glue at the south prevents a yellow tile from attaching there then, a contradiction. 
Consequently, $t_{\tt sbTF}(\north) = {\tt F}$ and $t_{\tt sbTT}(\north) = {\tt T}$. 
Now the glue assignment to the cyan tile types has been accomplished as follows: 

\vspace*{2mm}

\begin{center}
\begin{tikzpicture}
	\blueff{0}{0}; \node at (0, 0) {$t_{\tt sbFF}$}; 
	\blueft{2}{0}; \node at (2, 0) {$t_{\tt sbFT}$}; 
	\bluetf{4}{0}; \node at (4, 0) {$t_{\tt sbTF}$}; 
	\bluett{6}{0}; \node at (6, 0) {$t_{\tt sbTT}$}; 
\end{tikzpicture}
\end{center}

Now only the 2 {\tt DGNL}-white and 2 {\tt DGNL}-black tile types remain free from glues. 
See the pattern in Figure~\ref{fig:lb2-DGNL} (Right), where we find a {\tt DGNL}-white position next to the pattern {\tt Init}-red(F). 
The tile attaching there hence has east glue {\tt f} and south glue {\tt n}, no matter which type of white tile attaches below. 
Let us denote its type by $t_{\tt DGNLwF}$; then $t_{\tt DGNLwF}(\south) = {\tt n}$ and $t_{\tt DGNLwF}(\east) = {\tt f}$. 
As for the other type $t_{\tt DGNLwT}$, consider another instance of the template in Figure~\ref{fig:lb2-DGNL} rather with blue (T) at the top of the rightmost column. 
Then we obtain $t_{\tt DGNLwT}(\south) = {\tt n}$ and $t_{\tt DGNLwT}(\east) = {\tt t}$. 
Black analogues of these instances assign the {\tt DGNL}-black tile types $t_{\tt DGNLbF}$ and $t_{\tt DGNLbT}$ with glues partially as $t_{\tt DGNLbF}(\south) = t_{\tt DGNLbT}(\south) = {\tt v}$, $t_{\tt DGNLbF}(\east) = {\tt f}$, and $t_{\tt DGNLbT}(\east) = {\tt t}$. 
\vspace*{2mm}
\begin{center}
\begin{tikzpicture}
	\bluetile{0}{0}{}{}{}{}{}; \whitelowertriangle{0}{0}; \glsn{0}{0}; \glef{0}{0}; 
	\bluetile{2}{0}{}{}{}{}{}; \whitelowertriangle{2}{0}; \glsn{2}{0}; \glet{2}{0}; 
	\bluetile{4}{0}{}{}{}{}{}; \blacklowertriangle{4}{0}; \glsv{4}{0}; \glef{4}{0}; 
	\bluetile{6}{0}{}{}{}{}{}; \blacklowertriangle{6}{0}; \glsv{6}{0}; \glet{6}{0}; 
\end{tikzpicture}
\end{center}

On the tenth column from the right in Figure~\ref{fig:gadget2}, there is a {\tt DGNL}-white position. 
The tile attaching there is of type $t_{\tt DGNLwF}$, and hence, it is assigned with glues as $t_{\tt DGNLwF}(\north) = t_{\tt DGNLwF}(\west) = {\tt F}$. 
As for the other type $t_{\tt DGNLwT}$, see the fourth and fifth columns in Figure~\ref{fig:gadget_DGNLwFF}, and two {\tt DGNL}-white positions are found on them in total. 
At both of them, tiles of this type attach. 
See the {\tt CE} positions just above them (on the sixth row from the bottom). 
{\tt CE} tiles attaching there must have $t_{\tt DGNLwT}(\north)$ as their south glue. 
Since $t_{\tt CEfs}$ tiles cannot attach next to each other and the south glue of other {\tt CE} tile types is {\tt F}, we get $t_{\tt DGNLwT}(\north) = {\tt F}$. 
For the sake of directedness, its west glue must not be {\tt F}. 
It actually must be {\tt T} in order to enable a tile of this type to attach to the east of cyan positions.  
\vspace*{2mm}
\begin{center}
\begin{tikzpicture}
	\DGNLwF{0}{0}; \DGNLwT{2}{0};
	\bluetile{4}{0}{}{}{}{}{}; \blacklowertriangle{4}{0}; \glsv{4}{0}; \glef{4}{0}; 
	\bluetile{6}{0}{}{}{}{}{}; \blacklowertriangle{6}{0}; \glsv{6}{0}; \glet{6}{0}; 
\end{tikzpicture}
\end{center}

Let us shift our attention to the {\tt DGNL}-black tile types $t_{\tt DGNLbF}$ and $t_{\tt DGNLbT}$. 
On the fourth column in Figure~\ref{fig:gadget2}, there is a {\tt DGNL}-black position, at which a $t_{\tt DGNLbF}$ tile attaches. 
Recall that the first two columns in the figure are identical to the rightmost two columns of Figure~\ref{fig:gadget1}. 
Therefore, we can apply Property~\ref{propty:yellow_CEss} to determine $t_{\tt DGNLbF}(\north) = {\tt F}$. 
Its west glue is determined as $t_{\tt DGNLbF}(\west) = {\tt F}$ since to its west are found cyan tiles, whose west and east glues are the same, and then a {\tt Sat} tile, with east glue {\tt F}. 
As for the assignment of the other type, see Figure~\ref{fig:gadget_DGNLbTT}, in which there is a {\tt DGNL}-black position whose north neighbor is yellow. 
The type of the {\tt DGNL}-black tile is $t_{\tt DGNLbT}$, and $t_{\tt DGNLbT}(\north) = {\tt T}$. 
Its west glue cannot be {\tt F} for the sake of directedness. 
Since east glues of cyan tiles are either {\tt F} or {\tt T} and {\tt DGNL}-black positions only ever appear to the east of cyan, in order for $t_{\tt DGNLbT}$ tiles to attach, $t_{\tt DGNLbT}(\west) = {\tt T}$ must hold. 

\vspace*{2mm}
\begin{center}
\begin{tikzpicture}
	\DGNLwF{0}{0}; \DGNLwT{2}{0}; \DGNLbF{4}{0}; \DGNLbT{6}{0}; 
\end{tikzpicture}
\end{center}

Now that all the 21 tile types have been assigned with glues completely, we should distinguish {\tt F} and {\tt T} from {\tt n} or {\tt v}. 
Compare the {\tt DGNL}-white tile type whose west glue is {\tt F} and two cyan tile types whose west glue is {\tt F}. 
For the directedness, they imply ${\tt n} \neq {\tt F}$ and ${\tt n} \neq {\tt T}$. 
In the same way, comparing {\tt DGNL}-black tile type with the west glue {\tt F} with these cyan tile types distinguishes {\tt v} from {\tt F} or {\tt T}.  
The tile type set have now turned out to be isomorphic to those in the set $T_{\tt val}$ (Property~\ref{propty:isomorphic}). 
This concludes the proof of Theorem~\ref{thm:11PATS_NPhard}. 

	\section{Proof of technical lemmas}

In this section, we prove the two lemmas left unproven in the previous section. 

	\subsection{Proof of Lemma~\ref{lem:LB4}}
	\label{subsec:proof_lemLB4}

	The color pattern of the top row of {\tt LB4} is represented as 
	\[
		\mbox{{\tt Sat}-red(F)-$[\ce]^2 Y^3 [\ce] Y^2 [\ce]^c Y [\ce]^2$-{\tt Sat}},
	\]
	where $Y$ indicates a yellow position and $c$ is some constant. 
	The rightmost column of {\tt LB4} is represented from bottom to top as ${\rm red(F)}^2$-${\rm true(T)}^{2r-1}$-${\rm red(F)}^{2r-1}$-${\rm true(T)}^2$-{\tt Sat}, where $r$ is some constant. 
	Recall that at the beginning of Section~\ref{subsec:gadget}, we claimed that the constants $c$ and $r$ are set large enough for the sake of this proof. 
	In fact, we set $c = 25$ and $r = 13$ (i.e., $2r-1 = 25$). 
	Needless to say, their values are set independently of $\phi$. 

	We will prove that if only 3 cyan tile types $A, B, C$ are available for the RTAS, then they have to be assigned with glues as shown in Figure~\ref{fig:spurious_cyan3}.  
	Below, we focus on the cyan region of {\tt LB4}; hence, for instance, by ``top row,'' we refer to the top row of the cyan region, unless otherwise noted. 

	First we deny the possibility that their west glues are pairwise distinct or all the same. 
	Indeed, with the pairwise-distinctness, the RTAS cannot help but assemble the top row in one of the following ways: 
	\[
	\begin{cases}
		A A \cdots 	& \text{if $A(\east) = A(\west)$} \\
		A B A B \cdots 	& \text{if $A(\east) = B(\west)$ and $B(\east) = A(\west)$} \\
		A B B \cdots 	& \text{if $A(\east) = B(\west) = B(\east)$} \\
		A B C A B C \cdots & \text{if $A(\east) = B(\west)$, $B(\east) = C(\west)$, and $C(\east) = A(\west)$} \\
		A B C B C \cdots & \text{if $A(\east) = B(\west)$, $B(\east) = C(\west)$, and $C(\east) = B(\west)$} \\
		A B C C \cdots	& \text{if $A(\east) = B(\west)$, $B(\east) = C(\west) = C(\east)$} 
	\end{cases}
	\]
	or their analogues obtained by changing the roles of $A, B, C$. 
	As a result, the cyan region exposes a periodic sequence of glues of period at most 3 to the north. 
	Recall that at the point where this lemma is concerned, only 7 tile types remain uncolored, and hence, even if we draw all of them by {\tt CE}, we have only 8 {\tt CE} tile types. 
	Imagine the task for the RTAS to assemble the top row of {\tt LB4}, or specifically, its subpattern $[{\tt CE}]^c Y$. 
	If the sequence of exposed glues by cyan region is of period 1 (all the glues are the same), then it would need to hardcode the position of Y by tiling all the $c = 25$ {\tt CE} positions with tiles of distinct types, but there are only at most 8 {\tt CE} tiles. 
	Even with period 3, after 24 {\tt CE} positions, pumping would occur and yellow tile would never attach, a contradiction. 
	In this way, the choice of the value for $c$ makes it impossible for the RTAS to assemble the top row if the cyan region consisting of all but the rightmost two columns exposes a periodic sequence of glues of period at most 3 to the north, and the choice of the value for $r$ is motivated analogously by the assemblability of the rightmost column. 

	Likewise, their south glues are not pairwise-distinct due to the same problem occurring on the east with large $r$. 
	Therefore, $A(\west) = B(\west) = C(\west)$ must NOT hold; otherwise, the directedness of the RTAS would imply the contradictory pairwise-distinctness of their south glues. 
	This also shows that two cyan tile types are not enough, as either their west glues or their south glues would need to be distinct for the sake of directedness. 

	Having figured out that there is no choice but $A(\west) = B(\west) \neq C(\west)$, let $A(\west) = B(\west) = 0$ and $C(\west) = 1$ for some distinct glues 0, 1. 
	For the sake of directedness, $A(\south) \neq B(\south)$ must hold; let $A(\south) = a$ and $B(\south) = b$ for some distinct glues $a, b$. 
	W.l.o.g., we can assume $C(\south) = a$. 
	Let us denote their north and east glues as follows:
	\begin{center}
	\begin{tikzpicture}
		\bluetile{0}{0}{A}{$n_1$}{0}{$a$}{$e_1$};
		\bluetile{2}{0}{B}{$n_2$}{0}{$b$}{$e_2$};
		\bluetile{4}{0}{C}{$n_3$}{1}{$a$}{$e_3$};
	\end{tikzpicture}
	\end{center}
	Recall that already at the beginning of this proof, we have denied $n_1 = n_2 = n_3$ and $e_1 = e_2 = e_3$. 

	We claim that their north glues must be either $a$ or $b$. 
	Firstly, if $n_1 \neq a, b$, then any row but the topmost one cannot help but assemble with only $B$ and $C$ tiles. 
	With $n_2 = n_3$, the second topmost row exposes a sequence of all the same glues to the north, and hence, the top row would assemble periodically, a contradiction. 
	On the other hand, $n_2 \neq n_3$ forces the RTAS to assemble the rightmost column as $B B B \cdots$, $C C C \cdots$, $B C B C \cdots$, or $C B C B \cdots$ up to its second topmost position, and this is enough for contradiction. 
	Thus, $n_1$ is $a$ or $b$. 
	Secondly, if $n_2 \neq a, b$, then any row except the top assembles with only $A$ and $C$ tiles. 
	Since $A$ and $C$ have the same south glue, these rows assemble periodically either as $A \cdots A$, $A C A C \cdots$, $A C C \cdots$, or their analogues which begin with $C$. 
	With $n_1 = n_3$, the top row would assemble periodically in one of these ways, a contradiction. 
	On the other hand, $n_1 \neq n_3$ forces the rightmost column to be assembled with tiles of sole type up to the third topmost positions, a contradiction. 
	Finally, if $n_3 \neq a, b$, then the rightmost column would assemble periodically with only $A$ and $B$ tiles as $A \cdots A$, $A B A B \cdots$, $A B \cdots B$, or their analogues which rather begin with $B$, a contradiction. 
	Therefore, $\{n_1, n_2, n_3\} = \{a, b\}$. 
	Analogously, we can prove $\{e_1, e_2, e_3\} = \{0, 1\}$. 

	Now we will prove that the one in Figure~\ref{fig:spurious_cyan3} is the only one set of 3 cyan tile types with which a directed RTAS can self-assemble {\tt LB4} provided no other cyan tile types are available. 

	\begin{description}

	\item[Case 1: $e_1 = e_2 = 0$, $e_3 = 1$ (or $n_1 = a$, $n_2 = b$, $n_3 = a$):]
		In order not to assemble the rightmost column periodically, at least one $C$ tile must be placed on the column. 
		Since a $C$ tile cannot be adjacent to any tile of distinct type, this means that a row assembles with just $C$ tiles. 
		All of the rows above would be mono-type as well, a contradiction. 
		Informally speaking, we have shown that these tiles must not copy their west glues to the east faithfully. 
		Its vertical analogue is that they must not copy their south glues to the north faithfully, that is, not that $n_1 = a$, $n_2 = b$, and $n_3 = a$. 

	\item[Case 2: $e_1 = e_2 = 1$, $e_3 = 0$ (or $n_1 = b$, $n_2 = a$, $n_3 = b$):]
		In this case, the tile types are as follows. 
	\begin{center}
	\begin{tikzpicture}
		\bluetile{0}{0}{A}{$n_1$}{0}{$a$}{$1$};
		\bluetile{2}{0}{B}{$n_2$}{0}{$b$}{$1$};
		\bluetile{4}{0}{C}{$n_3$}{1}{$a$}{$0$};
	\end{tikzpicture}
	\end{center}
		Any row admits a $C$ tile at every other position. 
		The bottom row assembles either as $C [A/B] C [A/B] \cdots$ or as $[A/B] C [A/B] C \cdots$. 
		With $n_3 = b$, any row but the bottom one would assemble periodically as $B C B C \cdots$ or $C B C B \cdots$, a contradiction. 
		Thus, $n_3 = a$, and this implies $n_1 = b$ due to Case 1. 
		If $n_2 = b$, then to the north of both $A$ and $B$ tiles are tiles of type $B$. 
		This means that the second lowest row assembles as $C B C B \cdots$, and so would all the rows above, a contradiction. 
		Thus $n_2 = a$. 

	\begin{center}
	\begin{tikzpicture}
		\bluetile{0}{0}{A}{$b$}{0}{$a$}{$1$};
		\bluetile{2}{0}{B}{$a$}{0}{$b$}{$1$};
		\bluetile{4}{0}{C}{$a$}{1}{$a$}{$0$};
	\end{tikzpicture}
	\end{center}

		To the north of $A$-tile is always a tile of type $B$. 
		This and the fact that at every row, $C$-tiles appear at every other position imply that if at a row, an $A$-tile attaches, then the assembly of the row just above is the image of that of the current row under the swap of $A$ and $B$ (for instance, above the row $CACBCA$, the row $CBCACB$ assembles). 
		This means that if both an $A$-tile and a $B$-tile attach at the bottom row, then the rightmost column is either the alternation of $A$ and $B$ tiles or consists of only $C$ tiles, a contradiction. 
		If the bottom row assembles as an alternation of $B$ and $C$ tiles, then each of the rows above is either an alternation of $A$ and $C$ tiles or that of $B$ and $C$ tiles, and hence, the topmost row would expose a periodic sequence of glues to the north, a contradiction. 
		Even if the bottom row assembles as an alternation of $A$ and $C$ tiles, this contradiction would arise. 
	\end{description}

	Due to Cases 1 and 2 and the fact that their north glues must not be all the same, now we know that $n_1 \neq n_3$ is necessary. 

	\begin{description}
	\item[Case 3: $e_1 = 0$, $e_2 = e_3 = 1$ (or $n_1 = a$, $n_2 = n_3 = b$):]
	The tile types are as follows: 
	\begin{center}
	\begin{tikzpicture}
		\bluetile{0}{0}{A}{$n_1$}{0}{$a$}{$0$};
		\bluetile{2}{0}{B}{$n_2$}{0}{$b$}{$1$};
		\bluetile{4}{0}{C}{$n_3$}{1}{$a$}{$1$};
	\end{tikzpicture}
	\end{center}
	It is clear that any row assembles as $A^* B C^*$, $A^*$, or $C^*$. 
	Since $B$ tiles cannot get adjacent to each other, if any but the top row assembles as $A^*$ (resp.~$C^*$), then $n_1 = a$ (resp.~$n_3 = a$). 
	We claim that at the bottom row, a $B$ tile must attach. 
	Indeed, if it assembled as $A^*$ (resp.~$C^*$), then $n_1 = a$ (resp.~$n_3 = a$) and all rows above would assemble using tiles of a single type, a contradiction. 
	Thus, the bottom row assembles as $A^i B C^j$ for some $i, j$. 

	With $n_2 = b$, to the north of $B$ tile is always a $B$ tile, and hence, any row above would assemble as $A^i B C^j$, and all the east glues of the cyan region would be identical, a contradiction. 
	Thus, $n_2 = a$, and hence, either $n_1$ or $n_3$ must be $b$. 
	If $n_1 = b$, then no more than one $A$ tile can be used to assemble any row but the top one. 
	In other words, any row other than the top row assembles either as $ABC^*$, $BC^*$, or $C^*$. 
	Then the sequence of east glues of the cyan region would be either $1^*$ or $1^*0$, a contradiction. 
	Thus, $n_1 = a$, and the requirement $n_1 \neq n_3$, hence, implies $n_3 = b$. 
	\begin{center}
	\begin{tikzpicture}
		\bluetile{0}{0}{A}{$a$}{0}{$a$}{$0$};
		\bluetile{2}{0}{B}{$a$}{0}{$b$}{$1$};
		\bluetile{4}{0}{C}{$b$}{1}{$a$}{$1$};
	\end{tikzpicture}
	\end{center}
	This suggests that, to the north of a $C$ tile, a $B$ tile attaches. 
	The assembly of any but the top row cannot include more than one $C$ tile as $B$ tiles must be placed above but they cannot get next to each other horizontally. 
	It is, hence, either $A^* B C$, $A^* B$, or $A^*$. 
	Therefore, the bottom two rows assemble as: 
	\begin{minipage}{0.45\linewidth}
	\[
	\begin{array}{ccccc}
	A & \cdots & A & A & B \\
	A & \cdots & A & B & C 
	\end{array}
	\]
	\end{minipage}
	\begin{minipage}{0.05\linewidth}
	or
	\end{minipage}
	\begin{minipage}{0.45\linewidth}
	\[
	\begin{array}{ccccc}
	A & \cdots & A & A & A \\
	A & \cdots & A & A & A/B  
	\end{array}
	\]
	\end{minipage}
	\vspace*{2mm}
	
	Any row above but the top one, hence, assembles as $A^*$, and in particular, the second top row exposes a sequence of $a$ glues to the north. 
	As a result, the top row would assemble either as $A^*$ or $C^*$, and expose a (periodic) sequence of same glues to the north, a contradiction. 
	\end{description}

	Before proceeding to the remaining cases, we should note that the remaining possibilities of $(n_1, n_2, n_3)$ are one of the following: 
	\begin{enumerate}
	\item	$n_1 = n_2 = a$, $n_3 = b$. 
	\item	$n_1 = b$, $n_2 = n_3 = a$. 
	\item	$n_1 = n_2 = b$, $n_3 = a$. 
	\end{enumerate}

	\begin{description}

	\item[Case 4: $e_1 = 1$, $e_2 = 0$, $e_3 = 1$ (or $n_1 = n_2 = b$, $n_3 = a$):]
	The tile types are as follows: 
	\begin{center}
	\begin{tikzpicture}
		\bluetile{0}{0}{A}{$n_1$}{0}{$a$}{$1$};
		\bluetile{2}{0}{B}{$n_2$}{0}{$b$}{$0$};
		\bluetile{4}{0}{C}{$n_3$}{1}{$a$}{$1$};
	\end{tikzpicture}
	\end{center}

	As such, the assembly of any row is represented as a factor of $B^* A C^*$. 
	How should the top row assemble? 
	We claim that it must begin with at least two $B$ tiles. 
	Indeed, otherwise, it assembles either as $BACC^*$, $ACC^*$, $C^*$, $B^*A$, or $B^*$. 
	In any case, the sequence of glues exposed by the cyan region would be periodic at least up to the third rightmost column. 
	Hence, it must assemble as $B^i BBACC^j$ for some $i, j \ge 0$. 

	Focus on this subassembly $BBAC$. 
	We claim that the type of tile attaching to the south of the second $B$ is $A$. 
	Indeed, if it were $C$, then the north glue of $C$ is fixed to $b$, but then no tile could attach to the east of this $C$ tile because $A(\south) = a$ and a tile which could attach to the east of $C$ tile is of type $C$. 
	If it were rather $B$, then the north glue of $B$ is $b$ so that to its east, only an $A$ tile can attach. 
	Hence, $A$ is provided with the north glue $a$, and this fixes the north glue of $C$ to $b$ because they are known to have distinct north glues. 
	However, then the second top row could not assemble because no tile could attach to the south of $C$ of the subassembly. 
	Hence, the top two rows assemble as: 
	\[
	\begin{array}{cccccccc}
	B & \cdots & B & B & A & C & \cdots & C \\
	B & \cdots & B & A & C & C & \cdots & C \\
	\end{array}
	\]
	and they impose $A(\north) = B(\north) = b$ and $C(\north) = a$. 
	This means that to the south of $C$ tile, only a $C$ tile can attach. 
	As a result, the rightmost column would expose a periodic sequence of glues to the east, a contradiction. 
	\end{description}

	The analysis of Case 4 has denied the possibility that $n_1 = n_2 = b$ and $n_3 = a$. 
	Now then the glue $n_2$ has been fixed to $a$. 
	The glue $e_3$ has been also fixed to 0. 
	The tile types are: 
	\begin{center}
	\begin{tikzpicture}
		\bluetile{0}{0}{A}{$n_1$}{0}{$a$}{$e_1$};
		\bluetile{2}{0}{B}{$a$}{0}{$b$}{$e_2$};
		\bluetile{4}{0}{C}{$n_3$}{1}{$a$}{$0$};
	\end{tikzpicture}
	\end{center}

	\begin{description}

	\item[Case 5: $e_1 = 1$, $e_2 = e_3 = 0$ (or $n_1 = b$, $n_2 = n_3 = a$):]
	The tile types are as follows: 
	\begin{center}
	\begin{tikzpicture}
		\bluetile{0}{0}{A}{$n_1$}{0}{$a$}{$1$};
		\bluetile{2}{0}{B}{$a$}{0}{$b$}{$0$};
		\bluetile{4}{0}{C}{$n_3$}{1}{$a$}{$0$};
	\end{tikzpicture}
	\end{center}

	The assembly of any row is represented as a factor of $(ACB^*)^*$. 
	We claim that any row but the top or bottom assembles in such a way that 
	\begin{enumerate}
	\item	two $B$ tiles do not get next to each other; 
	\item	$ACAC$ never appears. 
	\end{enumerate}
	In other words, we claim that the assembly of any row but the top or bottom is a factor of $(ACB)^*$. 
	Recall the necessity of $n_1 \neq n_3$, that is, one of them is $a$ and the other is $b$. 
	The first condition is certified by observing that no row can expose consecutive two $b$ glues to the north. 
	As for the second, suppose that we found $ACAC$ on a row. 
	If $n_1$ is $b$, then $n_3$ is $a$, and to the north of $A$ tiles, $B$ tiles must attach as: 
	\[
	\begin{array}{cccc}
	B & @ & B & \\
	A & C & A & C
	\end{array}
	\]
	However, then no tile could attach at the position $@$.
	The other case of $n_3$ being $b$ leads us to the same contradiction. 
	The second condition has been thus certified.

	Since $B$ and $C$ tiles have the same east glue, on the second rightmost column, an $A$ tile must occur. 
	We focus on one of such $A$ tiles; below it is marked as \fbox{$A$}. 
	Due to the above condition, around \fbox{$A$}, the assembly is like $\cdots B A C B \fbox{$A$} C$. 

	Let us observe how tiles attach around; we consider only the subcase when $n_3 = b$ (i.e., $n_1 = a$); the other subcase $n_1 = b$ and $n_3 = a$ is essentially symmetric and has the same effect. 
	In this subcase, the type of a tile above $C$ tile is $B$. 
	The row above, if any, assembles as 
	\[
	\begin{array}{cccc}
	B & A & C & B \\
	C & B & \fbox{$A$} & C
	\end{array}
	\]
	The assembly of rows above proceed in this way as follows: 
	\[
	\begin{array}{cccc}
	C & B & A & C \\
	A & C & B & A \\
	B & A & C & B \\
	C & B & \fbox{$A$} & C
	\end{array}
	\]
	The rows below assemble in the same way as: 
	\[
	\begin{array}{cccc}
	C & B & \fbox{$A$} & C \\
	A & C & B & A \\
	B & A & C & B \\
	C & B & A & C 
	\end{array}
	\]
	As a result, the rightmost column would assemble periodically, a contradiction. 

	\end{description}

	Now then only the tile type set depicted in Figure~\ref{fig:spurious_cyan3} has remained valid. 
	Let us reproduce it here for the sake of arguments below: 
	\begin{center}
	\begin{tikzpicture}

	\bluetile{0}{0}{A}{$a$}{0}{$a$}{0};
	\bluetile{2}{0}{B}{$a$}{0}{$b$}{1};
	\bluetile{4}{0}{C}{$b$}{1}{$a$}{0};

	\end{tikzpicture}
	\end{center}
	Observe that the south neighbor of $B$ tile is always of type $C$. 
	This suggests that being assembled with tiles of these types, {\tt LB4} does not expose two consecutive 1 glues eastward. 
	This property plays an important role in proving the need of 4 tile types of color red(F) or blue(T) in order to assemble {\tt LB4} with cyan tiles of these 3 types. 

	What we actually prove is that with at most 3 red(F)/blue(T) tile types, the rightmost column of {\tt LB4}, consisting of {\tt F}'s and {\tt T}'s, cannot be assembled. 
	Suppose there were at most 3 red(F)/blue(T) tile types. 
	Then either there is a sole red(F) tile type with at most 2 blue(T) tile types, or there is a sole blue(T) tile type with at most 2 red(F) tile types. 

	Let us only show that the rightmost column cannot assemble in the first case, as the argument for 1 blue(T) tile type can follow the same steps at analogous indexes. 
	Let $t_{\tt F}$ be the red(F) tile type and $t_{\tt T1}, t_{\tt T2}$ be the blue(T) tile types. 
	At all red(F) positions, $t_{\tt F}$ tiles are to attach. 
	Hence, $t_{\tt F}(\north) = t_{\tt F}(\south)$. 
	See the $2d-1$ consecutive red(F) positions on this column. 
	Due to the above-mentioned property of east glues of cyan tiles, $t_{\tt F}$ tiles forming this portion receive glue 0 from the west. 
	Thus, $t_{\tt F}(\west) = 0$, and this demands $t_{\tt F}(\south)$ be different from $a$ or $b$; let $t_{\tt F}(\north) = t_{\tt F}(\south) = c$.  
	\begin{center}
	\begin{tikzpicture}
		\node[tile, fill=red, text=white] at (0, 0) {{\tt F}}; 
		\node at (0, 0.6) {$c$}; 
		\node at (-0.6, 0) {0}; 
		\node at (0, -0.6) {$c$}; 

		\node[tile, fill=blue, text=white] at (2, 0) {{\tt T1}};
		\node[tile, fill=blue, text=white] at (4, 0) {{\tt T2}};
	\end{tikzpicture}
	\end{center}

	See the lowest blue(T) position. 
	W.l.o.g., the type of tile attaching there is $t_{\tt T1}$. 
	Then $t_{\tt T1}(\south) = c$, and hence, $t_{\tt T1}(\west)$ must not be 0 for the directedness; since cyan tiles can expose only 0 or 1 to their east, $t_{\tt T1}(\west) = 1$. 
	Since a tile attaching at its north neighbor cannot receive a glue 1 from the west, its type cannot be $t_{\tt T1}$, that is, it is $t_{\tt T2}$. 
	Hence, $t_{\tt T2}(\west) = 0$, and this requires $t_{\tt T2}(\south)$ be distinct from $a, b, c$; let $t_{\tt T2}(\south) = d$. 
	\begin{center}
	\begin{tikzpicture}
		\node[tile, fill=red, text=white] at (0, 0) {{\tt F}}; 
		\node at (0, 0.6) {$c$}; 
		\node at (-0.6, 0) {0}; 
		\node at (0, -0.6) {$c$}; 

		\node[tile, fill=blue, text=white] at (2, 0) {{\tt T1}};
		\node at (2, 0.6) {$d$}; 
		\node at (1.4, 0) {1}; 
		\node at (2, -0.6) {$c$}; 

		\node[tile, fill=blue, text=white] at (4, 0) {{\tt T2}};
		\node at (4, 0.6) {}; 
		\node at (3.4, 0) {0}; 
		\node at (4, -0.6) {$d$}; 
	\end{tikzpicture}
	\end{center}
	The column has assembled from the bottom as $t_{\tt F} t_{\tt F} t_{\tt T1} t_{\tt T2}$. 
	Due to the lack of a third blue(T) tile type, the north of $t_{\tt T2}$ must be either $c$ or $d$. 
	If it were $d$, then the column assembles as $t_{\tt F}^2 t_{\tt T1} t_{\tt T2}^{2d-2}$, but then it still exposes glue $d$ to its north and only a $t_{\tt T2}$ tile would attach, a contradiction. 
	Otherwise, the column assembles as $t_{\tt F}^2 (t_{\tt T1} t_{\tt T2})^{d-1} t_{\tt T1}$ and even in this case, its north neighbor would be colored blue(T) by choice of odd number of consecutive blue(T) positions, a contradiction. 
\hfill\qquad $\Box$

	\subsection{Proof of Lemma~\ref{lem:lb3}}
	\label{subsec:proof_lemlb3}

	\begin{figure}[tb]
	\begin{minipage}{0.45\linewidth}
	\begin{center}
	\includegraphics{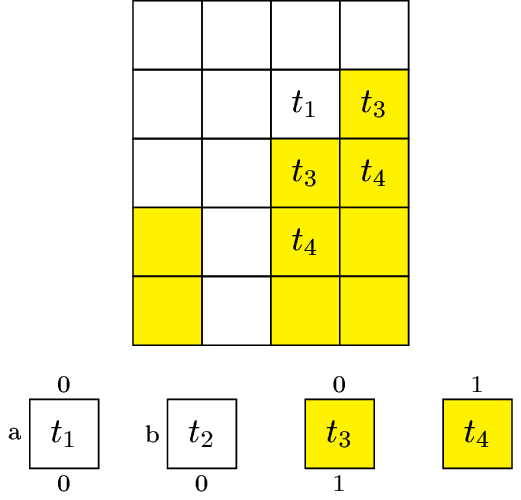}
	\end{center}
	\end{minipage}
	$\Longrightarrow$
	\begin{minipage}{0.45\linewidth}
	\begin{center}
	\includegraphics{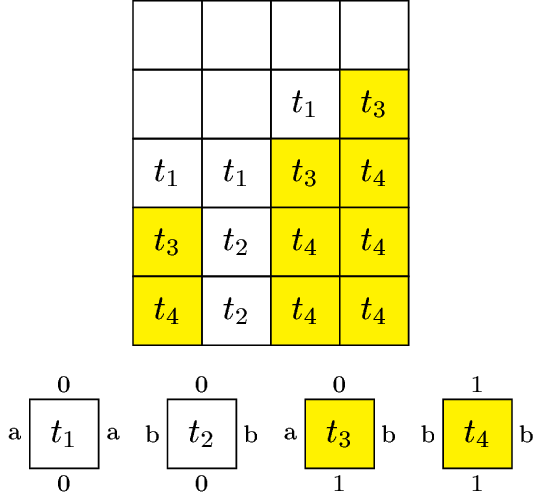}
	\end{center}
	\end{minipage}
	\caption{
		This is just a reproduction of Figure~\ref{fig:lb3-bc}. 
	}
	\label{fig:lb3-bc_reproduction}
	\end{figure}

Here, we prove Lemma~\ref{lem:lb3}. 
Since it refers to Figure~\ref{fig:lb3-bc}, we reproduce it here as Figure~\ref{fig:lb3-bc_reproduction}. 

\begin{figure}[tb]
\begin{minipage}{0.45\linewidth}
\begin{center}

\scalebox{1.0}{
\begin{tikzpicture}

\node[tile] at (0, 2.8) {};	\node[tile] at (0.7, 2.8) {};	\node[tile] at (1.4, 2.8) {};		\node[tile] at (2.1, 2.8) {};
\node[tile] at (0, 2.1) {};	\node[tile] at (0.7, 2.1) {};	\node[tile] at (1.4, 2.1) {\large $t_1$};	\node[yellowtile] at (2.1, 2.1) {};
\node[tile] at (0, 1.4) {};	\node[tile] at (0.7, 1.4) {};	\node[yellowtile] at (1.4, 1.4) {};	\node[yellowtile] at (2.1, 1.4) {};
\node[yellowtile] at (0, 0.7) {};\node[tile] at (0.7, 0.7) {};	\node[yellowtile] at (1.4, 0.7) {};	\node[yellowtile] at (2.1, 0.7) {};
\node[yellowtile] at (0, 0) {};	\node[tile] at (0.7, 0) {};	\node[yellowtile] at (1.4, 0) {};	\node[yellowtile] at (2.1, 0) {};

\end{tikzpicture}
}
\end{center}
\end{minipage}
\begin{minipage}{0.45\linewidth}
\begin{center}
\begin{tikzpicture}
	\whitetile{0}{0}{$t_1$}{0}{a}{0}{}; 
	\whitetile{2}{0}{$t_2$}{}{b}{0}{}; 
	\yellowtile{4}{0}{}{0}{c}{0}{}; 
\end{tikzpicture}

\end{center}
\end{minipage}
\caption{
	{\rm (Left)} A subpattern of {\tt GADGET} found to the north of {\tt LB4} in Figure~\ref{fig:gadget1}, where {\tt CE} positions are drawn simply by white for clarity.
	{\rm (Right) An imaginary set of two {\tt CE} tile types and one yellow tile type with which the left pattern could be assembled.} 
}
\label{fig:lb3-a}
\end{figure}
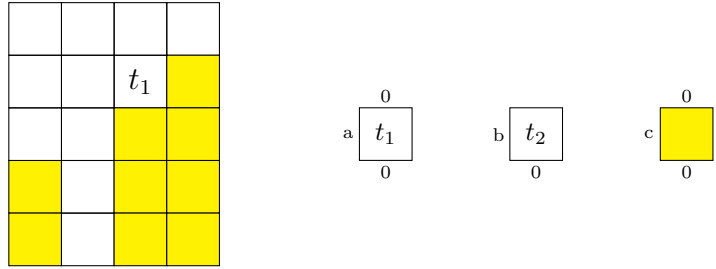

	We have already seen that any directed RTAS with at most 21 tile types needs at least two {\tt CE} tile types in order to self-assemble {\tt GADGET}. 
	If it has exactly two of them, say $t_1, t_2$, then as done in Lemma~\ref{lem:exactly2}, we can prove that $t_1(\west) \neq t_2(\west)$ and $t_1(\east) \neq t_2(\east)$, while $t_1(\south) = t_2(\south)$. 
	Let $t_1(\west) = a$ and $t_2(\west) = b$ for some distinct labels $a, b$, and let $t_1(\south) = t_2(\south) = 0$. 

	With three {\tt CE} tile types, the first statement of this lemma is trivial. 
	Hence, it suffices to prove that if the RTAS has exactly two {\tt CE} tile types $t_1, t_2$, then it must have at least 2 yellow tile types. 
	For the sake of contradiction, suppose that there were only one yellow tile type instead. 
	See Figure~\ref{fig:lb3-a}, where a subpattern of {\tt GADGET} is depicted, with {\tt CE} tiles being drawn rather just white for clarity. 
	W.l.o.g., the type of {\tt CE} tile at (3, 4) is $t_1$.
	Being self-stacked, the sole yellow tile type has the same north and south glues, and moreover, the glue is the same as the south glues of $t_1$ and $t_2$. 
	This means that its west glue must be distinct from $a$ or $b$; let it be $c$ (see Figure~\ref{fig:lb3-a} (Right)). 
	Then $t_1(\east) = c$, and this means that a $t_1$ tile cannot be adjacent to another $t_1$ tile horizontally, so the type of tile at (2, 4) is $t_2$. 
	However, then neither $t_1$ nor $t_2$ tiles can be at (1, 4) due to the east glue mismatch. 
	Therefore, if the RTAS has only 2 {\tt CE} tile types, it must have at least 2 yellow tile types $t_3$ and $t_4$. 

	Now let us prove the second statement of the lemma. 
	W.l.o.g., the type of yellow tile at (4, 4) is $t_3$. 
	As proved above, $t_3(\south)$ is not 0; let $t_3(\south) = 1$.  
	Not depending on the type of tile at (4, 5), $t_3(\north) = 0$. 
	This means that the type of tile at (4, 3) is not $t_3$ but $t_4$, and hence, let $t_4(\north) = t_3(\south) = 1$. 
	As shown in Figure~\ref{fig:lb3-bc_reproduction} (left), then the tiles at (1, 2) and (3, 3) are of type $t_3$ and their south neighbors are of type $t_4$. 
	Thus, $t_3(\east) = t_4(\west)$, and this glue is either $a$ or $b$ (see the positions (1, 2) and (2, 2)). 
	This means $t_4(\south) \neq 0$ or more strongly $t_4(\south) = 1$ because otherwise no yellow tile could attach to the south of a $t_4$ tile. 
	As illustrated in Figure~\ref{fig:lb3-bc_reproduction} (Right), any yellow column is to self-assemble in such a way that all but its topmost position is filled with $t_4$ tiles. 
	Since $t_3(\south) = t_4(\south) = 1$, their west glues must disagree, and this means that the white west neighbor of $t_3$ tile is always of type $t_1$ whereas that of $t_4$ tile is always of type $t_2$. 
	Now the resulting assembly of the pattern looks partially as depicted in Figure~\ref{fig:lb3-bc_reproduction} (Right).  
	In particular, $t_1$ tiles attach at both (1, 3) and (2, 3) and a $t_3$ tile attaches at (3, 3), and hence, $t_3(\west) = t_1(\east) = t_1(\west) = a$. 
	The assembly $t_4 t_2 t_4 t_4$ of the bottom row implies $t_4(\west) = t_4(\east) = t_2(\west) = t_2(\east) = b$. 
	Finally, $t_3(\east) = t_4(\west) = b$. 
	The glue assignment has been completed as shown in Figure~\ref{fig:lb3-bc_reproduction} (Right). 
\hfill\qquad$\Box$

	\bibliographystyle{plain}
	\bibliography{11pats_arxiv}

\end{document}